\documentclass[11pt]{article}
\usepackage{geometry} 
\usepackage{amsmath,amsthm,amsfonts,dsfont,mathtools,amssymb}
\usepackage{mathrsfs}
\usepackage{float}
\usepackage{thmtools}
\usepackage{thm-restate}
\usepackage{xcolor}
\usepackage{graphicx}
\usepackage{algorithmic,algorithm}
\usepackage{enumitem}
\setitemize{noitemsep,topsep=5pt,parsep=5pt,partopsep=5pt}
\usepackage{multirow}

\definecolor{darkgreen}{rgb}{0,0.5,0}
\usepackage{hyperref}
\hypersetup{
    unicode=false,          % non-Latin characters in Acrobat’s bookmarks
    colorlinks=true,        % false: boxed links; true: colored links
    linkcolor=violet,          % color of internal links (change box color with linkbordercolor)
    citecolor=teal,        % color of links to bibliography
    filecolor=magenta,      % color of file links
    urlcolor=cyan           % color of external links
}
\usepackage[capitalize, nameinlink]{cleveref}

\crefname{theorem}{Theorem}{Theorems}
\Crefname{lemma}{Lemma}{Lemmas}
\Crefname{figure}{Figure}{Figures}
\Crefname{claim}{Claim}{Claims}
\Crefname{observation}{Observation}{Observations}

\newtheorem{fact}{Fact}[section]

\newtheorem{lemma}{Lemma}[section]

\newtheorem{observation}{Observation}[section]

\newcommand{\silence}{\mathsf{silence}}
\newcommand{\collision}{\mathsf{collision}}

\newcommand{\cd}{\mathsf{CD}}
\newcommand{\nocd}{\mathsf{No}{\text -}\mathsf{CD}}

\newcommand{\binomial}{\operatorname{Binomial}}
\newcommand{\ignore}[1]{}

\newcommand{\Prob}{\operatorname{Pr}}

\newcommand{\poly}{{\operatorname{poly}}}

\newcommand{\dist}{\operatorname{dist}}
\newcommand{\ecc}{\operatorname{eccentricity}}

\newcommand{\ID}{\operatorname{ID}}

\newcommand{\LOCAL}{\mathsf{LOCAL}}

\newcommand{\CONGEST}{\mathsf{CONGEST}}

\newcommand{\sr}{\mathsf{SR}{\text -}\mathsf{comm}}
\newcommand{\minsr}{\mathsf{SR}{\text -}\mathsf{comm}^{\operatorname{min}}}
\newcommand{\maxsr}{\mathsf{SR}{\text -}\mathsf{comm}^{\operatorname{max}}}
\newcommand{\apxsr}{\mathsf{SR}{\text -}\mathsf{comm}^{\operatorname{apx}}}
\newcommand{\allsr}{\mathsf{SR}{\text -}\mathsf{comm}^{\operatorname{all}}}
\newcommand{\multisr}{\mathsf{SR}{\text -}\mathsf{comm}^{\operatorname{multi}}}

\title{The Energy Complexity of Diameter and Minimum Cut Computation in Bounded-Genus Networks}

\author{Yi-Jun Chang\\
National University of Singapore
}

\begin{document}
\date{}
\maketitle
\thispagestyle{empty}
\setcounter{page}{0}

\begin{abstract}
This paper investigates the energy complexity of distributed graph problems in multi-hop radio networks, where the energy cost of an algorithm is measured by the maximum number of awake rounds of a vertex. Recent works revealed that some problems, such as broadcast, breadth-first search, and maximal matching,  can be solved with energy-efficient algorithms that consume only $\poly \log n$ energy. However, there exist some problems, such as computing the diameter of the graph, that require $\Omega(n)$ energy to solve. To improve energy efficiency for these problems, we focus on a special graph class: bounded-genus graphs. We present algorithms for computing the exact diameter, the exact global minimum cut size, and a $(1
\pm\epsilon)$-approximate $s$-$t$ minimum cut size with $\tilde{O}(\sqrt{n})$ energy for bounded-genus graphs. Our approach is based on a generic framework that divides the vertex set into high-degree and low-degree parts and leverages the structural properties of bounded-genus graphs to control the number of certain connected components in the subgraph induced by the low-degree part.
\end{abstract}

\section{Introduction \label{sect.intro}}

We consider the multi-hop \emph{radio network} model~\cite{chlamtac1985broadcasting} of distributed computing, where a communication network is modeled as a graph $G=(V,E)$: Each vertex $v\in V$ is a computing device and each edge $\{u, v\} \in E$ indicates that $u$ and $v$ are within the transmission range of each other.
The graph topology of the underlying network $G$ is initially unknown to all devices, except that two parameters $n = |V|$ and $\Delta = \max_{v \in V} \deg(v)$ are global knowledge.

We assume that the communication proceeds in synchronized rounds. All devices agree on the same start time.
In each round, each device can choose to do one of the following three operations: (i) listen to the channel, (ii) transmit a message, or (iii) stay idle.
We do not allow a device to simultaneously transmit and listen, and we assume that there is no message size constraint. 

Each transmitting or idle device does not receive any feedback from the communication channel, so a transmitting device $u$ does not know whether its message is successfully received by any of its neighbors $N(u)$.
A listening device $v$ successfully receives a message from a transmitting device $u \in N(v)$ if $u$ is the only transmitting device in $N(v)$. If the number of transmitting devices in $N(v)$ is zero, then a listening device $v$ hears $\silence$. If the number of transmitting devices in $N(v)$ is greater than one, then the feedback that a listening device $v$ receives depends on the underlying model. In the $\nocd$ model (without collision detection), $v$ still hears $\silence$. In the $\cd$ model (with collision detection), $v$ hears $\collision$. All our algorithms presented in this paper work in the $\nocd$ model.

We assume that each device has access to an unlimited local random source.
We say that an event occurs \emph{with high probability} (w.h.p.) if the event occurs with probability $1 - 1/\poly(n)$.
If we let each vertex $v \in V$ locally assign themselves $O(\log n)$-bit identifiers $\ID(v)$, then they are distinct w.h.p., so  we may assume that each device  has a distinct identifier  of length  $O(\log n)$.

\paragraph{Complexity measures.}
\emph{Time} and \emph{energy} are the two main complexity measures of distributed algorithms in radio networks.
The {time complexity} of an algorithm is the number of rounds of the algorithm in the worst case. Unless otherwise stated, all our algorithms are \emph{Monte Carlo} in that they succeed w.h.p.
The {energy complexity} of an algorithm is  the maximum energy cost of a device in the worst case, where the energy cost of a device $v$ is the number of rounds that $v$ is non-idle.  
The motivation for studying energy complexity is that energy is a scarce resource in small battery-powered wireless devices, and such devices can save energy by entering a low-power sleep mode.

\subsection{Prior work}

Most of the early work on the energy complexity focused on \emph{single-hop} radio networks, which is the special case where $G=(V,E)$ is a complete graph.
Over the last two decades, there has been a long line of research to optimize the energy complexity of \emph{leader election} and its related problems in single-hop radio networks~\cite{BenderKPY16,bordim2000energy,caragiannis2005basic,ChangKPWZ17,Chang2021detLE,ChangSPAA2022,JurdzinskiKZ02cocoon,JurdziskiKZ03,JurdzinskiKZ02podc,kardas2013energy,kutylowski2003adversary,lavault2007quasi,nakano2000randomized}. 

This line of research was recently extended to multi-hop radio networks~\cite{ChangDHHLP18,Chang20bfs,dani2021wake,dani2022wake}.
Chang~et~al.~\cite{ChangDHHLP18} considered the  problem of \emph{broadcasting} a message from one device to all other devices in a multi-hop radio network.
They showed that broadcasting can be done in $\poly \log n$ energy. Specifically, they presented randomized broadcasting algorithms for $\cd$ and $\nocd$ using energy $O\left(\frac{\log n\log\log\Delta}{\log\log\log\Delta}\right)$ and $O(\log\Delta\log^2 n)$ w.h.p., respectively. They also proved that any algorithm transmitting a message from one endpoint to the other endpoint of an $n$-vertex path costs $\Omega(\log n)$ energy \emph{in expectation}. The lower bound applies even to the $\LOCAL$ model of distributed computing.

 Chang~et~al.~\cite{Chang20bfs} showed that \emph{breadth-first search} can be done w.h.p.~using $2^{O\left(\sqrt{\log n \log \log n}\right)}$ energy in $\nocd$. Their algorithm is based on a hierarchical clustering using the low-diameter decomposition algorithm of Miller, Peng, and Xu~\cite{miller2013parallel}. 
The energy complexity of breadth-first search was recently improved to $\poly \log n$ by Dani and Thomas~\cite{dani2022wake}.
Combining the polylogarithmic-energy breadth-first search algorithm with the diameter approximation algorithm of Roditty and Williams~\cite{RodittyW13}, an approximation $\tilde{D}$ of the diameter $D$ such that $\tilde{D} \in \left[\lfloor 2D/3\rfloor, D\right]$ can be computed with $\tilde{O}(\sqrt{n})$ energy w.h.p.~\cite{Chang20bfs}. The notation $\tilde{O}(\cdot)$ suppresses any $\poly \log n$ factor.

Dani~et~al.~\cite{dani2021wake} showed that a \emph{maximal matching} can be computed in $O(\Delta \log n)$ time and $O(\log \Delta \log n)$ energy w.h.p in $\nocd$. There exists a family of graphs such that these time and energy bounds are simultaneously optimal up to polylogarithmic factors.

\subsection{Our contribution}

Not all problems admit energy-efficient algorithms in multi-hop radio networks. It was shown in~\cite{Chang20bfs} that any algorithm that computes a $(1.5-\epsilon)$-approximation of the diameter requires $\tilde{\Omega}(n)$ energy w.h.p. The lower bound holds even on graphs with arboricity $O(\log n)$ and treewidth $O(\log n)$.

 \paragraph{Bounded-genus graphs.} To improve energy efficiency for diameter computation, we focus on the class of bounded-genus graphs.  The {\em genus} of a graph $G$ is the minimum number $g$ such that $G$ can be drawn on an oriented surface of $g$ handles without crossing. For example, planar graphs are the graphs with genus zero, and the graphs that can be drawn on a torus without crossing are the graphs with genus at most one.
A class of graphs is called \emph{bounded-genus} if the genus of all graphs in the class can be upper bounded by some constant $g = O(1)$.

\paragraph{Diameter.} We show that the diameter of the graph can be computed using $\tilde{O}(\sqrt{n})$ energy w.h.p.~in bounded-genus graphs.

 \begin{restatable}{theorem}{thmdiammain}\label{lem:diam-ub}
There is an algorithm that computes the diameter in $\tilde{O}(n^{1.5})$ time and $\tilde{O}(\sqrt{n})$ energy w.h.p.~for bounded-genus graphs in $\nocd$.
\end{restatable}

 Our approach is based on a generic framework that divides the vertex set into high-degree and low-degree parts.
 We then classify the connected components of the subgraph induced by the low-degree part into several types. We  will leverage the structural properties of bounded-genus graphs to upper-bound the number of connected components of one type.
 For the remaining connected components, we will design energy-efficient algorithms that extract all the necessary information from these connected components for the purpose of diameter computation.

\paragraph{Minimum cut.}  Our approach is sufficiently general so that it is applicable to other problems as well.
Using the same approach, we show that the exact global minimum cut size and a $(1\pm \epsilon)$-approximate $s$-$t$ minimum cut size can also be computed using $\tilde{O}(\sqrt{n})$ energy w.h.p.~in bounded-genus graphs.

 \begin{restatable}{theorem}{thmcutmain}\label{thm:mincut}
There is an algorithm that computes the minimum cut size in $\tilde{O}(n^{1.5})$ time and $\tilde{O}(\sqrt{n})$ energy w.h.p.~for bounded-genus graphs in $\nocd$.
\end{restatable}

 \begin{restatable}{theorem}{thmcutapxmain}\label{thm:stmincut}
There is an algorithm that computes a $(1 \pm \epsilon)$-approximate $s$--$t$ minimum cut size in $\tilde{O}(n^{1.5}) + \tilde{O}(\sqrt{n}) \cdot \epsilon^{-O(1)}$ time and $\tilde{O}(\sqrt{n} + \epsilon^{-O(1)})$ energy w.h.p.~for bounded-genus graphs in $\nocd$.
\end{restatable}

To complement these algorithmic results, we show that
any algorithm that computes the exact size of an $s$--$t$ minimum cut or a global minimum cut requires $\Omega(n)$ energy. The lower bound for the $s$--$t$ minimum cut holds even for planar bipartite graphs, so it is necessary that we consider approximation algorithms for this problem. These lower bounds apply to both $\nocd$ and $\cd$.

 \begin{restatable}{theorem}{thmLBa}\label{lem:learn-deg}
For any randomized algorithm that computes the $s$--$t$ minimum cut size of a planar bipartite graph w.h.p.~in $\cd$, the energy complexity of the algorithm is $\Omega(n)$. 
\end{restatable}
  
 \begin{restatable}{theorem}{thmLBb}\label{lem:learn-deg2}
For any randomized algorithm that computes the minimum cut size of a unit-disc graph w.h.p.~in $\cd$, the energy complexity of the algorithm is $\Omega(n)$. 
\end{restatable}

\subsection{Additional related work}

Klonowski and Pajak~\cite{KlonowskiP17} considered a variant of the model where only transmitting costs energy, and they showed that in $\nocd$, for any $1 \leq  \varphi \leq O(\log n /  \log \log n)$, broadcasting can be solved in $O((D+\varphi) n^{1/\varphi} \varphi)$ time using $O(\varphi)$ transmission per device.

There are numerous works studying energy-aware distributed computing in multi-hop networks from different perspectives.
In radio networks, the power of a signal received is proportional to $O(1/d^{\alpha})$, where $d$ is the distance to the sender, and $\alpha$ is a constant related to environmental factors. Kirousis~et~al.~\cite{KirousisKKP00} studied the optimization problem of assigning transmission ranges of devices subject to some connectivity and diameter constraints so as to minimize the total power consumption. See~\cite{Ambuhl05,Clementi01,TakagiK84} for related work.

There are several works~\cite{BerenbrinkHC07,EphremidesT90,SenH97} on the subject of reducing the number of rounds or transmissions required to complete a specific communication task. In the setting of known network topology, Gasieniec~et~al.~\cite{Gsieniec07} designed a randomized protocol for broadcasting in $O(D + kn^{1/(k-2)}\log^2 n)$ rounds such that each device transmits at most $k$ times.

The energy complexity has recently been studied  in  the well-known $\LOCAL$ and $\CONGEST$ models of distributed computing~\cite{augustine2022distributed,barenboim_et_al:LIPIcs.DISC.2021.10,ChatterjeeGP20,dufoulon2022sleeping,ghaffari2022average}. 

There is a large body of research on distributed graph algorithms in special graph classes such as planar networks, bounded-genus networks, or more broadly $H$-minor-free networks:  distributed approximation~\cite{
akhoondian2016local,
czygrinow2008fast,
lenzen2013distributed,
wawrzyniak2014strengthened}, low-congestion shortcuts and its applications~\cite{ghaffari2016distributed,
ghaffari2021low,
ghaffari2017near,
haeupler2018minor}, and other planar graph algorithms~\cite{li2019planar,parter2020distributed}.

\section{Tools\label{sect:tools}}
In this section, we present the basic tools that we use in our algorithms.

\subsection{Communication between two sets of vertices}

Let $\mathcal{S}$ and $\mathcal{R}$ be two  vertex sets that are not necessarily disjoint.
The task $\sr$~\cite{ChangDHHLP18} is defined as follows.
Each vertex $u \in \mathcal{S}$ holds a message $m_u$ that it wishes to transmit,
and each vertex $v \in \mathcal{R}$ wants to receive one message from vertices in $N^+(v) \cap \mathcal{S}$, where $N^+(v) = N(v) \cup \{v\}$ is the inclusive neighborhood of $v$. The message that $v \in \mathcal{R}$ receives can be any message in $N^+(v) \cap \mathcal{S}$.
In other words, the task $\sr$ requires that w.h.p.~for
each vertex $v \in \mathcal{R}$ with $N^+(v) \cap \mathcal{S} \neq \emptyset$,  there exists a vertex $u \in N^+(v) \cap \mathcal{S}$ such that $v$ receives a message $m_u$ from  $u$.
Several variants of $\sr$ are defined as follows.
\begin{description}
\item[All messages: $\allsr$.] The task $\allsr$ requires that each vertex $v \in \mathcal{R}$ receives the message $m_u$  for each $u \in N^+(v) \cap \mathcal{S}$ w.h.p. 
\item[Approximate sum: $\apxsr$.] Suppose the message $m_u$ sent from each vertex $u \in \mathcal{S}$ is an integer within the range $[W]$.
The task $\apxsr$ requires that each vertex $v \in \mathcal{R}$  computes a $(1 \pm \epsilon)$-factor approximation of the summation $\sum_{u \in N^+(v) \cap \mathcal{S}} m_u$ w.h.p.
\item[Minimum and maximum: $\minsr$ and $\maxsr$.] The message $m_u$ sent from each vertex $u \in \mathcal{S}$ contains a key $k_u$ from the key space $[K] = \{1, 2, \ldots, K\}$. For $\minsr$, it is required that w.h.p., each vertex $v \in \mathcal{R}$ with $N^+(v) \cap \mathcal{S} \neq \emptyset$ receives a message $m_u$  from a vertex $u \in N^+(v) \cap \mathcal{S}$ such that $k_u = \min_{u' \in N^+(v) \cap \mathcal{S}} k_{u'}$. The task $\maxsr$ is defined analogously by replacing minimum with maximum.
\item[Multiple messages: $\multisr$.]  Consider the setting where each vertex $u \in \mathcal{S}$ holds a set of messages $\mathcal{M}_u$.
For each message $m$, all vertices holding the same message $m$ have access to some shared random bits associated with $m$.
We assume that for each $v \in \mathcal{R}$, the number of distinct messages in $\bigcup_{u \in N^+(v) \cap \mathcal{S}} \mathcal{M}_u$ is upper bounded by a number $M$ that is known to all vertices.
 The task $\multisr$ requires that each vertex  $v \in \mathcal{R}$ receives all distinct messages in $\bigcup_{u \in N^+(v) \cap \mathcal{S}} \mathcal{M}_u$ w.h.p.
\end{description}

\cref{tab:my_label} summarizes the time and energy complexities of our algorithms for these tasks. For $\allsr$, the parameter $\Delta'$ can be any known upper bound on $|\mathcal{S} \cap N(v)|$, for each $v \in \mathcal{R}$. For example, we may set  $\Delta' = \Delta$ if no better upper bound is known.
The proofs for these results are left to \cref{sect.sr}.
Note that for the special case where $\mathcal{S} \cap \mathcal{R} = \emptyset$ and $|\mathcal{R} \cap N(u)| \leq 1$ for each $u \in \mathcal{S}$, the time complexity of  $\minsr$ and $\maxsr$ can be improved to $O(\log K \log \Delta \log n)$.

\begin{table}[]
\begin{center}
\begin{tabular}{llll}
{\bf Task} & {\bf Time} & {\bf Energy}                                        &                                                                     \\ \cline{1-3}
\multicolumn{1}{|l|}{$\sr$}      & \multicolumn{1}{l|}{$O(\log \Delta \log n)$}                       & \multicolumn{1}{l|}{$O(\log \Delta \log n)$}                         &                                                                     \\ \cline{1-3}
\multicolumn{1}{|l|}{$\allsr$}   & \multicolumn{1}{l|}{$O( \Delta' \log n)$}                          & \multicolumn{1}{l|}{$O( \Delta' \log n)$}                            & $\Delta' \geq |\mathcal{S} \cap N(v)|, \, \forall v \in \mathcal{R}$ \\ \cline{1-3}
\multicolumn{1}{|l|}{$\minsr$}   & \multicolumn{1}{l|}{\multirow{2}{*}{$O(K \log \Delta \log n)$}}    & \multicolumn{1}{l|}{\multirow{2}{*}{$O(\log K \log \Delta \log n)$}} &                                                                     \\ \cline{1-1}
\multicolumn{1}{|l|}{$\maxsr$}   & \multicolumn{1}{l|}{}                                              & \multicolumn{1}{l|}{}                                                &                                                                     \\ \cline{1-3}
\multicolumn{1}{|l|}{$\apxsr$}   & \multicolumn{1}{l|}{$O(\epsilon^{-6} \log W \log \Delta \log n)$} & \multicolumn{1}{l|}{$O(\epsilon^{-6}\log W\log \Delta \log n)$}    &                                                                     \\ \cline{1-3}
\multicolumn{1}{|l|}{$\multisr$} & \multicolumn{1}{l|}{$O(M \log \Delta \log^2 n)$}                   & \multicolumn{1}{l|}{$O(M \log \Delta \log^2 n)$}                     &                                                                     \\ \cline{1-3}
\end{tabular}
\end{center}
    \caption{The time and energy complexities of $\sr$  and its variants.}
    \label{tab:my_label}
\end{table}

\subsection{Communication via a good labeling\label{sect.label}}

A {\em good labeling} is a vertex labeling $\mathcal{L}: V(G) \mapsto \{0, \ldots, n-1\}$ such that each vertex $v$ with $\mathcal{L}(v) > 0$ has a neighbor $u$ with $\mathcal{L}(u)=\mathcal{L}(v)-1$~\cite{ChangDHHLP18}.
A vertex $v$ is called a {\em layer-$i$} vertex if $\mathcal{L}(v) = i$.
If there is a unique layer-0 vertex $r$, then $\mathcal{L}$ represents a tree $T$ rooted at $r$, so we call $r$  the \emph{root} of  $\mathcal{L}$. Since a vertex might have multiple choices of the parent, the tree $T$ is not  unique in general. The following lemma was proved in~\cite{ChangDHHLP18}.

\begin{lemma}[\cite{ChangDHHLP18}] \label{lemma:labeling}
A good labeling $\mathcal{L}$ with a unique layer-0 vertex can be constructed in $O(n \log \Delta \log^2 n)$ time and $O(\log \Delta \log^2 n)$ energy w.h.p. 
\end{lemma}

The following lemma shows that a good labeling allows the vertices in the graph to broadcast messages in an energy-efficient manner. 

\begin{lemma} \label{lemma:bc}
Suppose that we are given a good labeling $\mathcal{L}$ with a unique layer-0 vertex.  Then we can achieve the following.
\begin{enumerate}
\item It takes  $O(n \Delta \log n)$ time  and $O(\Delta \log n)$ energy for each vertex to broadcast a message to the entire network w.h.p.
\item It takes $O(nx \log \Delta \log^2 n)$ time and $O(x \log \Delta \log^2 n)$ energy for $x$ vertices to broadcast messages to the entire network w.h.p.
\end{enumerate}
\end{lemma}
\begin{proof}
Let $r$ be the root of $\mathcal{L}$.
For the first task, consider the following algorithm. We relay the message of each vertex to the root $r$ using the following {\em convergecast} algorithm.
For $i = n-1$ down to $1$, do  $\allsr$ with $\mathcal{S}$ being the set of all layer-$i$ vertices and $\mathcal{R}$ being the set of all layer-$(i-1)$ vertices. For each execution of  $\allsr$, each vertex in $\mathcal{S}$ transmits not only its message but also all other messages that it has received so far. Although we perform $\allsr$  $n-1$ times, each vertex only participates at most twice. By \cref{lemma:allsr}, the cost of the convergecast algorithm is  $O(n \Delta \log n)$ time  and $O(\Delta \log n)$ energy.

At the end of the convergecast algorithm, the root $r$ has gathered all messages sent during the algorithm. After that, the root $r$ then broadcasts this information to all vertices via the following {\em divergecast} algorithm.
For $i = 0$ to $n - 2$, do  $\sr$ with $\mathcal{S}$ being the set of all layer-$i$ vertices and $\mathcal{R}$ being the set of all layer-$(i+1)$ vertices.
Similarly, although we perform $\sr$ for $n-1$ times, each vertex only participates at most twice. By \cref{lemma:sr},  the cost of the divergecast algorithm is  $O(n \log \Delta \log n)$ time  and $O(\log \Delta \log n)$ energy. At the end of the divergecast algorithm, all vertices have received all messages.

For the rest of the proof, we consider the second task. Let $X$ be the set of $x$ vertices that attempt to broadcast a message.
We solve this task similarly in two steps:
\begin{itemize}
    \item We first do a convergecast, using $\multisr$ with $M=x$, to gather all $x$ messages to the root. By \cref{lemma:multisr},  $\multisr$  costs $O(x\log \Delta \log^2 n)$ time and energy, so the convergecast costs $O(nx \log \Delta \log^2 n)$ time and $O(x \log \Delta \log^2 n)$ energy.
    \item After that, we do a divergecast based on $\sr$ to broadcast these messages from root to everyone. The divergecast  costs  $O(n \log \Delta \log n)$ time  and $O(\log \Delta \log n)$ energy.
\end{itemize}
In order to use $\multisr$, the initial holder of each message $m$ needs to first generate a sufficient number of random bits and attach them to the message. These random bits serve as the shared randomness associated with the message $m$, which is needed in the definition of $\multisr$.
\end{proof}

We make the following observation.

\begin{observation}\label{obs:gathering}
There is an algorithm that lets all vertices learn the entire graph topology in  $O(n \Delta \log n)$ time  and $O(\Delta \log n)$ energy w.h.p.
\end{observation}
\begin{proof}
 We first let each vertex $v$ learn the list of identifiers in $N(v)$ by doing $\allsr$ with $\mathcal{S}=\mathcal{R}=V$, where the message of each vertex $v$ is $\ID(v)$. By \cref{lemma:allsr}, this step takes $O(\Delta \log n)$ time and energy. After that, we apply \cref{lemma:labeling} to construct a good labeling with a unique layer-0 vertex, and then we apply \cref{lemma:bc}(1) to let all vertices learn the entire network topology by having each $v$ broadcasting $\ID(v)$ and the list of identifiers in $N(v)$. This step takes  $O(n \Delta \log n)$ time  and $O(\Delta \log n)$ energy.  
\end{proof}

\section{Graph partitioning}\label{sect:partitioning}
In this section, we consider a classification of the connected components of the subgraph induced by the low-degree vertices in a bounded-genus graph. Our algorithms, which will be presented in subsequent sections, make use of the classification.

Let $G = (V,E)$ be any bounded-genus graph.
Let $V_H$ be the set of vertices that have degree at least $\sqrt{n}$.
Let $V_L = V \setminus V_H$. Since bounded-genus graphs have arboricity $O(1)$, we have $|E| = O(n)$, which implies $|V_H| = O(\sqrt{n})$. 

From now on, we assume $|V_H| \geq 1$, since otherwise $G$ has maximum degree $\Delta \leq \sqrt{n}$, in which case we can already solve {\em all} problems using $O(n \Delta \log \Delta \log n) = \tilde{O}(n^{1.5})$ time and $O(\Delta \log \Delta \log n) = \tilde{O}(\sqrt{n})$ energy by learning the entire graph topology using the algorithm of \cref{obs:gathering}.

Given a set of vertices $S$, we write $G[S]$ to denote the subgraph of $G$ induced by $S$ and write $G^+[S]$ to denote the subgraph  of $G$ induced by all edges that have {\em at least one} endpoint in $S$.
We classify the connected components of $G[V_L]$ into three types. 
\begin{description}
\item[Type 1.] A connected component $S$ of $G[V_L]$ is of type-1 if $|S| < \sqrt{n}$ and $|\bigcup_{w \in S} N(w) \cap V_H| = 1$. For each vertex $u \in V_H$, we write $C(u)$ to denote the set of type-1 components $S$ such that $\bigcup_{w \in S} N(w) \cap V_H = \{u\}$.
\item[Type 2.] A connected component $S$ of $G[V_L]$ is of type-2 if $|S| < \sqrt{n}$ and $|\bigcup_{w \in S} N(w) \cap V_H| = 2$. For each pair of two distinct vertices $\{u,v\} \subseteq V_H$, we write $C(u,v)$ to denote the set of type-2 components $S$ such that $\bigcup_{w \in S} N(w) \cap V_H = \{u,v\}$.
\item[Type 3.] A connected component $S$ of $G[V_L]$ is of type-3 if it is neither of type-1 nor of type-2.
\end{description}
A connected component $S$ of $G[V_L]$ is of type-3 if $|S| \geq \sqrt{n}$ or $|\bigcup_{w \in S} N(w) \cap V_H| \geq 3$. The number of type-3 components $S$ with $|S| \geq \sqrt{n}$ is clearly at most $|V|/\sqrt{n} = \sqrt{n}$. Utilizing the assumption that $G$ is a bounded-genus graph, we  show that the number of type-3 components with $|\bigcup_{w \in S} N(w) \cap V_H| \geq 3$ is also $O(\sqrt{n})$.

\begin{lemma}\label{lem:euler}
Let $G=(V,E)$ be a bipartite graph of genus at most $g$. Let  $V = X \cup Y$ be the bipartition of $G$.
If $\deg(v) \geq 3$ for each $v \in X$, then $|X| \leq 2|Y| + 4(g-1)$.
\end{lemma}
\begin{proof}
Consider any embedding of $G$ into a surface of genus $g$, and let $F$ be the set of faces of the embedded graph.
In a bipartite graph, each face has at least four edges, and each edge appears in at most two faces, so  $|E| \geq 2|F|$.
Combining this inequality with Euler's polyhedral formula $|V|-|E|+|F| \geq 2 - 2g$, we obtain that \[2|V| - |E| \geq 4(1-g).\]
Since $\deg(v) \geq 3$ for each $v \in X$, we have $|E| \geq 3|X|$, so
\[2|V| - |E| = 2(|X| +|Y|) - |E| \leq 2(|X| +|Y|) - 3|X| = 2|Y| - |X|.\]
Combining these upper and lower bounds of $2|V| - |E|$, we obtain that $2|Y| - |X| \geq 4(1-g)$, so $|X| \leq 2|Y| + 4(g-1)$, as required.
\end{proof}

 \cref{lem:euler} is precisely the reason that our algorithms only apply to bounded-genus graphs and do not work on an arbitrary $H$-minor-free graph.
 Consider a complete bipartite graph with the bipartition $X \cup Y$ such that $|Y| = 3$. Such a graph does not contain $K_5$ as a minor, regardless of the size of $X$. Therefore,  $K_5$-minor-freeness does not allow us to upper bound $|X|$ by any function of $|Y|$. 
Therefore, the bounded-genus requirement in \cref{lem:euler} cannot be relaxed to $H$-minor-freeness for an arbitrary $H$.

\begin{lemma}\label{lem:type3}
If $G$ is a bounded-genus graph, then the number of type-3 components is $O(\sqrt{n})$.
\end{lemma}
\begin{proof}
A connected component $S$ of $G[V_L]$ is of type-3 if $|S| \geq \sqrt{n}$ or $|\bigcup_{w \in S} N(w) \cap V_H| \geq 3$.
As discussed earlier, the number of type-3 components $S$ with $|S| \geq \sqrt{n}$ is at most $\sqrt{n}$, so we just need to prove that
the number of type-3 components $S$ with $|\bigcup_{w \in S} N(w) \cap V_H| \geq 3$ is also $O(\sqrt{n})$.
Consider a bipartite graph $G^\ast=(V^\ast, E^\ast)$ with the bipartition $V^\ast = X^\ast \cup Y^\ast$ defined as follows.
\begin{itemize}
    \item $X^\ast$ is the set of all  type-3 components $S$ such that $|\bigcup_{w \in S} N(w) \cap V_H| \geq 3$.
    \item $Y^\ast = V_H$.
    \item For each component $S \in X^\ast$ and each vertex $v \in Y^\ast$, $\{S,v\} \in E^\ast$ if $v \in \bigcup_{w \in S} N(w)$.
\end{itemize}
Alternatively, $G^\ast$ can be constructed from $G$ by the following steps.
\begin{itemize}
    \item Remove all type-1, type-2, and   type-3 components $S$ with $|\bigcup_{w \in S} N(w) \cap V_H| \leq 2$.
    \item For each type-3 component $S$  with $|\bigcup_{w \in S} N(w) \cap V_H| \geq 3$, contract $S$ into a vertex.
\end{itemize}

As $G^\ast$ can be obtained from $G$ via a sequence of edge contractions and vertex removals, $G^\ast$ is a bounded-genus graph.  Observe that $\deg(S) \geq 3$ for each $S \in X^\ast$ in $G^\ast$, so we may apply 
 \cref{lem:euler}, which shows that the number $|X^\ast|$ of type-3 components $S$ such that $|\bigcup_{w \in S} N(w) \cap V_H| \geq 3$ satisfies
$|X^\ast| \leq 2|Y^\ast| + O(1) = 2|V_H| + O(1) = O(\sqrt{n})$.
\end{proof}

We write $G_H$ to denote the graph defined by the vertex set $V_H$ and the edge set $\{\{u,v\} \, : \, |C(u,v)| > 0\}$.
The following observation is useful.

\begin{observation}\label{lem:type2}
If $G$ is a bounded-genus graph, then $G_H$ is also a bounded-genus graph, so the number of edges in $G_H$ is $O(\sqrt{n})$ and there exists an edge orientation of $G_H$ such that each vertex has outdegree $O(1)$.
\end{observation}
\begin{proof}
The graph $G_H$ can be obtained from $G$ via a sequence of edge contractions and vertex removals, so $G_H$ is a bounded-genus graph. As bounded-genus graphs have arboricity $O(1)$,  so the number of edges in $G_H$ is at most linear in the number of vertices in $G_H$, which is $O(\sqrt{n})$, and we can orient the edges of $G_H$ in such a way that each vertex has outdegree $O(1)$.
\end{proof}

\section{Diameter\label{sect.planar}}
In this section,   we show that for bounded-genus graphs, the diameter can be computed using $\tilde{O}(\sqrt{n})$ energy. 
We begin with discussing the high-level proof idea.
First of all, using \cref{lemma:bc}, learning the entire graph topology of the subgraph induced by $V_H$ and all type-3 components is doable using $\tilde{O}(\sqrt{n})$ energy. Intuitively, this is due to the following facts: (i) $|V_H| = O(\sqrt{n})$, (ii) $\deg(v) = O(\sqrt{n})$ for each $v \in V_L$, and (iii) the number of type-3 components is $O(\sqrt{n})$.

The main difficulty in the diameter computation is dealing with type-1 and type-2 components.  For example, a vertex $u \in V_H$ can be connected to  $\Theta(n)$ type-1 components in that $|C(u)| = \Theta(n)$. Since we aim for an algorithm with energy complexity $\tilde{O}(\sqrt{n})$,  throughout the entire algorithm, $u$ can only receive messages from at most $\tilde{O}(\sqrt{n})$ components in $C(u)$. The challenge here is to show that the diameter can still be calculated with a limited amount of information about type-1 and type-2 components and show that such information can be extracted in an energy-efficient manner in the radio network model.

We will define a set of parameters of type-1 and type-2 components and  show that with  these parameters, the exact value of the diameter can be calculated. Based on this result, we will define a subgraph $G^\star$ of $G$ such that the diameter of $G$ equals the diameter of $G^\star$, and then we will design an energy-efficient algorithm to learn the graph topology of $G^\star$.
In the subsequent discussion, we write $\ecc(u,S)$ to denote $\max_{v \in S} \dist(u,v)$. By default, all distances are measured in the underlying network $G$. We use subscripts to describe distances that are measured in a vertex set, an edge set, or a subgraph.

\paragraph{Parameters for type-1 components.}
We first consider the type-1 components in $C(u)$, for  any vertex $u \in V_H$.
\begin{description}
\item [{$(A_i[u], a_i[u])$.}] Let $A_1[u]$ be a component $S \in C(u)$ that maximizes $\ecc(u,S)$, and let $A_2[u]$ be a component $S \in C(u) \setminus \{A_1[u]\}$ that maximizes $\ecc(u,S)$.
For $i \in \{1,2\}$, we write $a_i[u] = \ecc(u,A_i[u])$.

\item [{$(B[u], b[u])$.}] Let $B[u]$ be a component $S \in C(u)$ that maximizes $\max_{s, t \in S \cup \{u\}} \dist(s,t)$,
and we write $b[u] = \max_{s, t \in B[u] \cup \{u\}} \dist(s,t)$.
\end{description}

In the above definitions, ties can be broken arbitrarily if there are multiple choices.
Some of the above definitions become undefined when $|C(u)|$  is too small. For example, if $|C(u)| = 1$, then $A_2[u]$ and $a_2[u]$ are undefined. In such a case, we set these parameters to their default values: zero or an empty set.
For example, if $|C(u)| = 1$, then we set $A_2[u] = \emptyset$ and $a_2[u] = 0$.

For each $u \in V_H$, any path connecting a vertex in $\bigcup_{S \in C(u)} S$  to a vertex outside of $\bigcup_{S \in C(u)} S$ must pass through  vertex $u$, so the amount of information we can afford to extract from $\bigcup_{S \in C(u)} S$ is limited. Intuitively, for the purpose of calculating the diameter, we only need  the following information from $\bigcup_{S \in C(u)} S$:
\begin{itemize}
\item The longest distance between two vertices in $\bigcup_{S \in C(u)} S \cup \{u\}$, which is $\max\{b[u], a_1[u]+a_2[u]\}$.
\item The longest distance between $u$ and a vertex in $\bigcup_{S \in C(u)} S$, which is $a_1[u]$.
\end{itemize}
Regardless of the size of $C(u)$, we only need to learn  $a_1[u]$, $a_2[u]$, and $b[u]$ from the components of $C(u)$. Later we will show that these parameters can be learned efficiently via $\maxsr$.

\paragraph{Parameters for type-2 components.}
Next, we consider the type-2 components in $C(u,v)$, for any two distinct vertices $u, v \in V_H$.
\begin{description}
\item [{$(R[u,v], r[u,v])$.}]   Let $R[u,v]$ be a component $S \in C(u,v)$ that minimizes $\dist_{G^+[S]}(u,v)$, and we write $r[u,v] = \dist_{G^+[R[u,v]]}(u,v)$.
    In other words, $R[u,v]$  is  a component that contains a shortest path between $u$ and $v$, among all $u$--$v$ paths via the vertices in $\bigcup_{S \in C(u,v)} S$.

\item [{$(A_i^k[u,v], a_i^k[u,v])$.}]
For each component $S \in C(u,v)$, we write $S^{u,k}$ to denote the set of vertices $\{w \in S \, : \, \dist_{G^+[S]}(w,v) - \dist_{G^+[S]}(w,u) \geq k\}$.
In other words, $S^{u,k}$ is  the set of all vertices in $S$ whose distance to $u$ in  $G^+[S]$ is shorter than that to $v$ by {\em at least} $k$.

Let $A_1^k[u,v]$ be a component $S \in C(u,v)$ that maximizes $\ecc_{G^+[S]}(u, S^{u,k})$, and
let $A_2^k[u,v]$ be a component $S \in C(u,v) \setminus \{A_1^k[u,v]\}$ that maximizes $\ecc_{G^+[S]}(u, S^{u,k})$.
We write $a_i^k[u,v] = \ecc_{G^+[A_i^k[u,v]]}(u,A_i^k[u,v])$.
We only consider $k \in \{-\sqrt{n}, \ldots, \sqrt{n}\}$.

\item [{$(B^l[u,v], b^l[u,v])$.}]
For a component $S \in C(u,v)$, we write $G^{l}[S]$ to denote the graph resulting from adding to $G^+[S]$ a path of length $l$ connecting $u$ and $v$, and we write $\phi^{l}(S)$ to denote the maximum value of  $\dist_{G^{l}[S]}(s,t)$ among all pairs of vertices $s, t \in S \cup \{u,v\}$.
A useful observation here is that if  $\dist_{V \setminus S}(u,v) = l$, then $\phi^{l}(S)$ equals the maximum value of $\dist_G(s,t)$ among all pairs of vertices $s, t \in S \cup \{u,v\}$.

Let $B^l[u,v]$ be a component $S \in C(u,v) \setminus \{ R[u,v] \}$ that maximizes $\phi^{l}(S)$,
and we write $b^l(u,v) = \phi^{l}(B^l[u,v])$.
We only consider $l \in \{1, \ldots, \sqrt{n}\}$.
\end{description}

Similar to the parameters of type-1 components, all the above parameters are set to their default values if they are undefined.
Note that the definitions of  $a_i^k[u,v]$ and  $A_i^k[u,v]$ are {\em asymmetric} in the sense that we might have $a_i^k[u,v] \neq a_i^k[v,u]$ and  $A_i^k[u,v] \neq A_i^k[v,u]$. All remaining parameters for type-2 components are symmetric.

We briefly explain how the above parameters can be  used in the diameter calculation.
Let $P=(s, \ldots, t)$ be an $s$--$t$ shortest path in $G$ whose length equals the diameter.
There are three possible ways that $P$  intersects the vertex set $\bigcup_{S \in C(u,v)} S$.
\begin{itemize}
\item Suppose the two endpoints $s$ and $t$ are within $G^+[S]$, for a component $S \in C(u,v)$. In this case, if $\dist_{V \setminus S}(u,v) = l$, then the length of $P$ equals $\phi^{l}(S) = b^l[u,v]$.
\item Suppose there is a subpath $P' = (u, \ldots, v)$ of $P$ whose intermediate vertices are all in $\bigcup_{S \in C(u,v)} S$. In this case, the length of $P'$ equals $r[u,v]$.
\item  Suppose there is a component $S \in C(u,v)$ such that $s \in S$ and $t \notin S \cup \{u,v\}$. Suppose $u$ is the first vertex of $P$ that is not in $S$.
Consider the subpath $P' = (s, \ldots, u)$  of $P$. If $\dist(t,u) - \dist(t,v) = k$, then we must have  $s \in S^{u,k}$, since otherwise $\dist(s,v) + \dist(v,t)$ is smaller than the length of $P$, violating the assumption that $P$ is an $s$--$t$ shortest path. If $t \notin A_1^k[u,v]$, then 
the length of $P'$ equals $a_1^k[u,v]$.
If $t \in A_1^k[u,v]$, then the length of $P'$ equals $a_2^k[u,v]$.
\end{itemize}

Intuitively, the above discussion shows that the parameters described above capture all the necessary information needed to be extracted from the type-2 components for the purpose of diameter computation.
We have $O(\sqrt{n})$ parameters for each $C(u,v)$. We will later show that all these parameters can  be learned using $O(\sqrt{n})$ energy.

\paragraph{The graph $G^\star$.}
We define  $G^\star$ as the subgraph induced by the union of (i) $V_H$, (ii) all type-3 components,
(iii) $A_1[u]$, $A_2[u]$, and $B[u]$, for all $u \in V_H$, and
(iv) $A_i^k[u,v]$, $A_i^k[v,u]$, $B^l[u,v]$, and $R[u,v]$, for all pairs of distinct vertices $\{u,v\} \subseteq V_H$,
$i \in \{1,2\}$,   $k \in \{-\sqrt{n}, \ldots, \sqrt{n}\}$, and  $l \in \{1, \ldots, \sqrt{n}\}$. 
In the subsequent discussion, we prove that the diameter of $G$ equals the diameter of $G^\star$, so the task of computing the diameter of $G$ is reduced to learning the topology of $G^\star$.
We will show that the following two statements are correct.
\begin{description}
\item [(S1)] For each pair of vertices $\{s,t\}$ in  graph $G^\star$, we have $\dist_G(s,t) = \dist_{G^\star}(s,t)$. 
\item [(S2)] For each pair of vertices $\{s,t\}$ in  graph $G$,
there exists a pair of vertices $\{s',t'\}$ in  graph $G^\star$
satisfying $\dist_G(s,t) \leq  \dist_{G}(s',t')$.
\end{description}
These statements imply that $G$ and $G^\star$ have the same diameter.
We first prove that (S1) is true.

\begin{lemma}\label{lem:cal-aux}
For any two vertices $s$ and $t$ in $G^\star$, we have  $\dist_{G}(s,t) = \dist_{G^\star}(s,t)$.
\end{lemma}
\begin{proof}
We choose $P$ to be an $s$--$t$ path in $G$ whose length is $\dist_{G}(s,t)$ that uses the minimum number of vertices not in $G^\star$.
If $P$ is entirely in $G^\star$, then we are done.
For the rest of the proof, we assume that $P$ is not entirely in $G^\star$.
Then  $P$ contains a subpath  $P' = (u, \ldots, v)$ whose intermediate vertices are all within a type-2 component $S \in  C(u, v)$ that is not included to $G^\star$.
By the definition of $R[u,v]$, the length of $P'$ is at least $r[u,v]$, which is the shortest path length between $u$ and $v$ via $R[u,v]$.
Therefore, replacing $P'$ with a shortest $u$--$v$ path in $R[u,v]$, which is entirely in $G^\star$, does not increase the path length.
This contradicts our choice of $P$.
Hence $P$ is entirely in $G^\star$.
\end{proof}

\begin{lemma}\label{lem:cal-aux3}
Let $S$ be a type-1 or type-2 component that is not included in $G^\star$.
Let $s$ be any vertex in $S$.
Let $t$ be any vertex in $G$ that does not belong to $G^+[S]$.
Then there exists a vertex $s'$ in $G^\star$  such that $\dist_{G}(s',t) \geq \dist_{G}(s,t)$.
\end{lemma}
\begin{proof}
Let $P$ be an $s$--$t$ shortest path in $G$.
Suppose that $S \in C(u)$ is of type-1.
Because $S$ is not included in $G^\star$, we must have $|C(u)| \geq 3$, so both $A_1[u] \neq S$ and $A_2[u] \neq S$ are not $\emptyset$.
Let $i \in \{1,2\}$ be an index such that $t$ is not in $A_i[u]$.
Consider the subpath $\tilde{P} = (s, \ldots, u)$  of $P$.
By the definition of  $a_i[u]$ and  $A_i[u]$, the length of $\tilde{P}$ is at most $a_i[u]$, and there exists a vertex $s' \in A_i[u]$ such that the length of the shortest path between $s'$ and $u$ equals $a_i[u]$. Thus, we have
$$\dist_{G}(s',t) = \dist_{G}(s',u) + \dist_{G}(u,t) \geq \dist_{G}(s,u) + \dist_{G}(u,t) = \dist_{G}(s,t).$$

Next, consider the case that $S \in C(u,v)$ is of type-2.
The path $P$ must contain at least one of $u$ and $v$.
Without loss of generality, assume that $u$ is the first vertex of $P$ that is not in $S$, so there is a subpath $\tilde{P} = (s, \ldots, u)$  of $P$ such that all vertices in $\tilde{P}$ other than $u$ are in $S$. The length of $P$ equals $\dist_{G^+[S]}(s,u) + \dist_{G}(u,t)$.

Let $k = \dist_{G^+[S]}(s,v) - \dist_{G^+[S]}(s,u)$, so $S^{u,k} \supseteq \{s\} \neq \emptyset$. 
Since $S$ is not of type-3, $|S| < \sqrt{n}$, so $k \in \{-\sqrt{n}, \ldots, \sqrt{n}\}$.
Because $S$ is not included in $G^\star$,  both $A_1^k[u,v] \neq S$ and $A_1^k[u,v] \neq S$ are not $\emptyset$.
At least one of $A_1^k[u,v]$ and $A_2^k[u,v]$ does not contain $t$. We choose $S' = A_i^k[u,v]$ as any one of them that does not contain $t$. We choose $s' \in S'$ as a vertex such that $\dist_{G^+[S']}(s',u) = a_i^k[u,v]$ and $\dist_{G^+[S']}(s',v) - \dist_{G^+[S']}(s',u) \geq k$. The existence of such a vertex $s'$ is guaranteed by the definition of  $A_i^k[u,v]$.

Our plan is to show that  (i) $a_i^k[u,v] + \dist_{G}(u,t) \geq \dist_{G}(s,t)$ and (ii) $\dist_{G}(s',t) = a_i^k[u,v] + \dist_{G}(u,t)$. Combining these two inequalities give us the desired result: $\dist_{G}(s',t) \geq \dist_{G}(s,t)$.

\paragraph{Proof of (i).} By the definition of  $A_i^k[u,v]$, we must have \[\dist_{G^+[S']}(s',u) = a_i^k[u,v]  \geq \dist_{G^+[S]}(s,u),\]  so we have \[a_i^k[u,v] + \dist_{G}(u,t) \geq \dist_{G^+[S]}(s,u) + \dist_{G}(u,t) = \dist_{G}(s,t).\]

\paragraph{Proof of (ii).}
Suppose that (ii) is not true. Then any shortest path between $s'$ and $t$ must contain a subpath $P' = (s', \ldots, v)$ such that $u$ is not in $P'$, and so we have:
$$\dist_{G}(s',t) = \dist_{G^+[S']}(s',v) + \dist_{G}(v,t) < \dist_{G^+[S']}(s',u) + \dist_{G}(u,t).$$
Combining this inequality with the known fact $\dist_{G^+[S']}(s',v) - \dist_{G^+[S']}(s',u) \geq k$, we have:
$$\dist_{G}(u,t) - \dist_{G}(v,t) > \dist_{G^+[S']}(s',v) - \dist_{G^+[S']}(s',u) \geq k,$$
which implies that $\dist_{G}(v,t) < \dist_{G}(u,t) - k$~($\star$).
We calculate an upper bound of $\dist_{G}(s,t)$:
\begin{align*}
\dist_{G}(s,t) &\leq \dist_{G^+[S]}(s,v) + \dist_{G}(v,t)\\
&= (k + \dist_{G^+[S]}(s,u)) + \dist_{G}(v,t) &\text{by definition of $k$.}\\
&< (k + \dist_{G^+[S]}(s,u)) + (\dist_{G}(u,t) - k) &\text{by ($\star$).}\\
&= \dist_{G^+[S]}(s,u) + \dist_{G}(u,t).
\end{align*}
This contradicts the assumption that $P$ is a shortest path between $s$ and $t$ in $G$, as the length of $P$ equals $\dist_{G^+[S]}(s,u) + \dist_{G}(u,t)$.
\end{proof}

The following lemma shows that (S2) is true.

\begin{lemma}\label{lem:cal-aux2}
For any two vertices $s$ and $t$ in graph $G$,
there exist two vertices $s'$ and $t'$ in  graph $G^\star$
such that $\dist_G(s,t) \leq  \dist_{G}(s',t')$.
\end{lemma}
\begin{proof}
If both  $s$ and $t$ are already in $G^\star$, then we are done by setting $s' = s$ and $t' = t$.
In the subsequent discussion,  we focus on the case that at least one of $s$ and $t$ is not in $G^\star$. 
By symmetry, we assume that $s$ is not in $G^\star$, so there is a type-1 or a type-2 component $S$ that is not included in $G^\star$ such that $s \in S$.

\paragraph{Case 1: $t$ belongs to $G^+[S]$.} 
If $S \in C(u)$ for some $u \in V_H$, then there exist two vertices $s'$ and $t'$ in the component $B[u] \in C(u)$ such that $\dist_G(s',t') = b[u] \geq \dist_G(s,t)$ by the definition of $B[u]$.

The remaining case is where $S \in C(u,v)$ for some $u,v \in V_H$. Let $l = \dist_{V \setminus S}(u,v)$.
We observe that $l \leq r[u,v]$.  The reason is that the existence of a component $S \neq R[u,v]$  guarantees that $R[u,v] \neq \emptyset$, which implies that \[l = \dist_{V \setminus S}(u,v) \leq  \dist_{G^+[R[u,v]]}(u,v) = r[u,v],\] as $G^+[R[u,v]]$ is a subgraph of $G[V \setminus S]$.

Since $R[u,v]$ is of type-2, we have $r[u,v] \leq |R[u,v]|+1  \leq \sqrt{n}$, so $l \in \{1, \ldots, \sqrt{n}\}$
Consider the component $B^l[u,v] \in C(u,v)$.
We observe that $l = \dist_{V \setminus B^l[u,v]}(u,v)$, since the shortest $u$--$v$ path length via $R[u,v]$ is at most the length of any $u$--$v$ path via $S$ or $B^l[u,v]$, by our choice of $R[u,v]$. More precisely, we have:
\[l = \dist_{V \setminus S}(u,v) = \dist_{V}(u,v) =
\dist_{V \setminus B^l[u,v]}(u,v),\]
as the above discussion implies that including $S$ and excluding  $B^l[u,v]$ in the subscript does not change the shortest $u$--$v$ path length. Here we use the fact that $B^l[u,v] \neq R[u,v]$, which is due to the definition of $B^l[u,v]$.

Since $l = \dist_{V \setminus B^l[u,v]}(u,v)$, by the definition of $B^l[u,v]$, there exist two vertices $s'$ and $t'$ in $G^+[B^l[u,v]]$ such that $\dist_{G}(s',t') \geq \dist_{G}(s,t)$, since otherwise we would have selected $B^l[u,v] = S$.

\paragraph{Case 2: $t$ does not belong to $G^+[S]$.}  
We apply \cref{lem:cal-aux3} to find a vertex $s'$ in $G^\star$ such that $\dist_G(s,t) \leq  \dist_{G}(s',t)$.
If $t$ is already in $G^\star$, then we are done.
Otherwise, there is a type-1 or a type-2 component $S'$ that is not included in $G^\star$ such that $t \in S'$. There are two sub-cases.
\begin{itemize}
    \item Suppose $s'$  belongs to $G^+[S']$. Then we may apply the same argument for Case~1 above to find two vertices $s''$ and $t''$ in $G^\star$ such that $\dist_G(s,t) \leq  \dist_{G}(s',t) \leq \dist_G(s'', t'')$.
    \item  Suppose $s'$  does not belong to $G^+[S']$. Then we may apply \cref{lem:cal-aux3} again to find a vertex $t'$ in $G^\star$ such that $\dist_G(s,t) \leq  \dist_{G}(s',t) \leq \dist_{G}(s',t')$.
\end{itemize}
In both sub-cases, we find two vertices in $G^\star$ whose distance in $G$ is at least
 $\dist_G(s,t)$.
\end{proof}

We are now ready to prove that the diameter of $G$ equals the diameter of $G^\star$.

\begin{lemma}\label{lem:cal-diam}
The diameter of $G$ equals the diameter of $G^\star$.
\end{lemma}
\begin{proof}
\cref{lem:cal-aux} shows that (S1) is true.
 \cref{lem:cal-aux2} shows that (S2) is true.
 These two results together imply that $G$ and $G^\star$ have the same diameter. Statement (S1) implies that the diameter of $G^\star$ is at most the diameter of $G$. For the other direction, let $s$ and $t$ be two vertices in $G$ such that $\dist(s,t)$ equals the diameter of $G$. By (S2), there exist two vertices $s'$ and $t'$ in $G^\star$
such that $\dist_G(s,t) \leq  \dist_{G}(s',t')$. By (S1), $\dist_G(s',t') = \dist_{G^\star}(s',t')$, so the diameter of $G^\star$ is at least the diameter of $G$.
\end{proof}

\subsection{Learning the topology of \texorpdfstring{$G^\star$}{Gstar}\label{sect:algo-details}}

By \cref{lem:cal-diam}, the task of computing the diameter of a bounded-genus graph $G$ is reduced  to computing the diameter of $G^\star$.
In this section, we show that  all  vertices can learn the graph topology of $G^\star$ using $\tilde{O}(\sqrt{n})$ energy.

Recall that $G_H$ is the graph defined by the vertex set $V_H$ and the edge set $\{\{u,v\} \, : \, |C(u,v)| > 0\}$. By \cref{lem:type2}, we know that $E(G_H) = O(\sqrt{n})$ and
 there exists an assignment $F: E(G_H) \mapsto V_H$ mapping each pair $\{u,v\} \in E(G_H)$ to one vertex in $\{u,v\}$ such that each $w \in V_H$ is mapped to at most $O(1)$ times.
Let $\mathcal{A}'$ be any deterministic centralized algorithm that finds such an assignment $F$, and we fix $F^\star$ to be the outcome of $\mathcal{A}'$ on the input $G_H$. If each vertex $v \in V$ already knows the graph $G_H$, then $v$ can locally calculate $F^\star$ by simulating $\mathcal{A}'$.

To learn  $G^\star$, we will let each vertex $u \in V$ learn the following information:
\begin{description}
\item[Basic information $\mathcal{I}_0(u)$.]
 For each vertex $u \in V$, $\mathcal{I}_0(u)$ contains the following information: (i) whether $u \in V_H$ or $u \in V_L$, (ii) the list of vertices in $N(u) \cap V_H$, and (iii)  the set of all pairs $\{u',v'\} \in E(G_H)$.

If $u$ is in a connected component $S$ of $G[V_L]$, then $\mathcal{I}_0(u)$ contains the following additional information:
(i) the list of vertices in $S$, and (ii) the topology of the subgraph $G^+[S]$.

\item[Information about type-1 components $\mathcal{I}_1(u)$.]
For each $u \in V_H$,  $\mathcal{I}_1(u)$ contains the  graph topology of $G^+[S']$, for each $S' = A_1[u]$, $A_2[u]$, and $B[u]$.

\item[Information about type-2 components $\mathcal{I}_2(u)$.]
For each $u \in V_H$, $\mathcal{I}_2(u)$ contains the following information.
For each pair $\{u,v\}  \in E(G_H)$ such that $F^\star(\{u,v\}) = u$,
$\mathcal{I}_2(u)$ includes the  graph topology of $G^+[S']$,  for each $S' = A_i^k[u,v]$, $A_i^k[v,u]$, $B^l[u,v]$, and $R[u,v]$, for each
$i \in \{1,2\}$,  $k \in \{-\sqrt{n}, \ldots, \sqrt{n}\}$, and  $l \in \{1, \ldots, \sqrt{n}\}$.
\end{description}

Information $\mathcal{I}_0(u)$ contains the graph topology of $G_H$, allowing each vertex $u$ to calculate $F^\star$ locally. 
Note that $\mathcal{I}_1(u)$ and $\mathcal{I}_2(u)$ contain nothing if $u \in V_L$.
The following lemma shows that the graph topology of $G^\star$ can be learned efficiently given that  each vertex $u \in V$ already knows  $\mathcal{I}_0(u)$, $\mathcal{I}_1(u)$, and $\mathcal{I}_2(u)$.

\begin{lemma}\label{lem:Gstar}
Given that each $u \in V$ already knows $\mathcal{I}_0(u)$, $\mathcal{I}_1(u)$, and $\mathcal{I}_2(u)$,
using $\tilde{O}(n^{1.5})$ time and  $\tilde{O}(\sqrt{n})$ energy, we can let all vertices in $G$ learn the graph topology of $G^\star$ w.h.p.
\end{lemma}
\begin{proof}
To learn $G^\star$, it suffices to know the following information: (i)  $\mathcal{I}_1(u)$ and $\mathcal{I}_2(u)$ for each $u \in V_H$,  (ii) the graph topology of $G^+[S]$ for each type-3 component $S$, and (iii) the graph topology of the subgraph induced by $V_H$.
For each type-3 component $S$, let $r_S$ be the smallest ID vertex in $S$.
In view of the above, to let each vertex learn the topology of $G^\star$, it suffices to let the following $O(\sqrt{n})$ vertices broadcast the following information:
\begin{itemize}
\item For each $u \in V_H$, $u$ broadcasts $\mathcal{I}_1(u)$, $\mathcal{I}_2(u)$, and the list of vertices $N(u) \cap V_H$, which is contained in $\mathcal{I}_0(u)$.

\item For each $u \in V_L$ such that $u = r_S$ for a type-3 component $S$, $u$ broadcasts the graph topology of $G^+[S]$. Note that each vertex $u \in V_L$ can decide locally using the information in $\mathcal{I}_0(u)$ whether or not $u$ itself is $r_S$ for a type-3 component $S$.
\end{itemize}
 Since $|V_H| = O(\sqrt{n})$ and the number of type-3 components is also $O(\sqrt{n})$ by \cref{lem:type3},
 the number of vertices that has a message to broadcast is $O(\sqrt{n})$.
 We run the algorithm of \cref{lemma:labeling} to find a good labeling $\mathcal{L}$ of $G$, and then we  use \cref{lemma:bc}(2) with $x = O(\sqrt{n})$ to let the above $O(\sqrt{n})$  vertices broadcast their information. This can be done in time  $\tilde{O}(n^{1.5})$ and   energy  $\tilde{O}(\sqrt{n})$. After that, all vertices know the graph topology of $G^\star$.
\end{proof}

Next, we consider the task of learning the basic information $\mathcal{I}_0(u)$.

\begin{lemma}\label{lem:HL}
Using $\tilde{O}(\sqrt{n})$ time and energy, we can let all vertices $v\in V$ learn the following information w.h.p.
\begin{itemize}
    \item Each $v \in V$ learns  whether $v \in V_H$ or $v \in V_L$.
    \item If $v \in V_H$, then $v$ also learns the list of vertices in $N(v) \cap V_H$.
    \item If $v \in V_L$, then $v$ also learns the two lists of vertices $N(v) \cap V_L$ and  $N(v) \cap V_H$.
\end{itemize}
\end{lemma}
\begin{proof}
First, we run $\apxsr$ with $W=1$, $\epsilon = 1/2$, $\mathcal{S}=\mathcal{R}=V$, and $m_u = 1$, for each $u \in \mathcal{S}$. This step lets each $v \in V$ estimate $\deg(v)$ up to a factor of 2. This step costs $\poly\log n$ time, by \cref{lemma:apxsr}.

After that, we run $\allsr$ with $\mathcal{S} = V$ and  $\mathcal{R}$ being the set of all vertices $v$ whose estimate of $\deg(v)$ is at most $2 \sqrt{n}$. The message $m_v$ for each vertex $v$ is $\ID(v)$, and we use the bound  $\Delta' = 4\sqrt{n}$ for $\allsr$.
Recall that $V_L$ is the set of vertices of degree at most $\sqrt{n}$, so we must have $V_L \subseteq \mathcal{R}$.
The algorithm of $\allsr$ allows each vertex $v \in \mathcal{R}$ to calculate $\deg(v)$ precisely.
 Therefore, after this step, each vertex $v \in V$  has enough information to decide whether $v \in V_H$ or $v \in V_L$. Furthermore, if $v \in V_L$, then $v$ knows the list of all vertices $N(v)$. This step takes  $\tilde{O}(\sqrt{n})$ time, by \cref{lemma:allsr}.

In order for each vertex to learn all the required vertex lists, we  run $\allsr$ again with the following parameters:
 $\mathcal{S} = V_H$,  $\mathcal{R} = V$, and the message $m_v$ for each vertex $v \in \mathcal{S}$ is its $\ID(v)$.
This time we may use the bound  $\Delta' = \sqrt{n} \geq |V_H|$.  After the algorithm of $\allsr$, each vertex $v \in V$ knows  the list of vertices in $N(v) \cap V_H$.   For each $v \in V_L$, since $v$ already knows the list of all vertices $N(v)$, it can locally calculate the list  $N(v) \cap V_L$.  This step also takes  $\tilde{O}(\sqrt{n})$ time.
\end{proof}

\begin{lemma}\label{lem:S}
Using $\tilde{O}(n^{1.5})$ time and  $\tilde{O}(\sqrt{n})$ energy, we can let all vertices $v$ in all connected components $S$ of $G[V_L]$ learn (i) the vertex set $S$ and (ii) the graph topology of $G^+[S]$ w.h.p.
\end{lemma}
\begin{proof}
First, we apply \cref{lem:HL} to let all vertices $v \in V_L$ learn the two lists  $N(v) \cap V_L$ and  $N(v) \cap V_H$.
To let all vertices learn the required information in the lemma statement, it suffices to let each  vertex $v \in S$ broadcast the two lists $N(v) \cap V_L$ and  $N(v) \cap V_H$ to all other vertices in $S$, for all connected components $S$ of $G[V_L]$.

We do the above broadcasting task in parallel, for all connected components $S$ of $G[V_L]$.
We use \cref{lemma:labeling} to let each component $S$ compute a good labeling, and then we use \cref{lemma:bc}(1) to let each vertex $v \in S$  broadcast the two lists  $N(v) \cap V_L$ and  $N(v) \cap V_H$ to all other vertices in $S$.
Recall that the degree of any vertex in $V_L$ is less than $\sqrt{n}$, so the algorithm of \cref{lemma:bc}(1) costs
$\tilde{O}(n^{1.5})$ time and  $\tilde{O}(\sqrt{n})$ energy.
\end{proof}

For each connected component $S$ of $G[V_L]$, at the end of the algorithm of \cref{lem:S}, each vertex $w \in S$ is able to determine the type of $S$. If $S$ is of type-1, $w$ knows the vertex $u \in V_H$ such that $S \in C(u)$. If $S$ is of type-2,  $w$  knows the two vertices $u, v \in V_H$ such that $S \in C(u,v)$. Given such information, in the following lemma, we design an algorithm for learning the topology of $G_H$.

\begin{lemma}\label{lem:pairs}
Suppose that each vertex in each type-2 component $S$ already knows (i) the vertex set $S$ and (ii) the graph topology of $G^+[S]$.
Using $\tilde{O}(n^{1.5})$ time and $\tilde{O}(\sqrt{n})$ energy, all vertices in the graph can learn the set of all pairs $\{u,v\} \in E(G_H)$ w.h.p.
\end{lemma}
\begin{proof}
First of all, we let all vertices in $V_H$ agree on a fixed ordering $V_H = \{v_1, \ldots, v_{|H|}\}$ as follows.
We use \cref{lemma:labeling} to compute a good labeling of $G$, and then we use \cref{lemma:bc}(2) with $x = \sqrt{n}$ to let each vertex $v \in V_H$ broadcast $\ID(v)$. After that, we may order $V_H = \{v_1, \ldots, v_{|H|}\}$ by increasing ordering of $\ID$. This step takes $\tilde{O}(n^{1.5})$ time and $\tilde{O}(\sqrt{n})$ energy.

Next, we consider the task of letting each $u \in V_H$ learn the list of all $v \in V_H$ such that $C(u,v) \neq \emptyset$.
We solve this task by $|V_H|$ invocations of $\sr$.
Given a type-2 component $S \in C(u,v)$,  we define $z_{u,S}$ as the smallest-ID vertex in $N(v) \cap S$.
The vertex $z_{u,S}$ will be responsible for letting $v$  know that  $C(u,v) \neq \emptyset$.
For $i = 1$ to $|V_H|$, we do an $\sr$ with $\mathcal{R} = V_H$  and  $\mathcal{S}$ being the set of all vertices  $z_{v_i,S}$ such that $S$ is a type-2 component with $v_i \in G^+[S]$. Observe that a vertex $u \in V_H$ receives a message during the $i$th iteration if and only if $C(u, v_i)\neq \emptyset$, i.e., $\{u, v_i\} \in E(G_H)$.
By \cref{lemma:sr}, this step takes $|V_H| \cdot \poly \log n = \tilde{O}(\sqrt{n})$ time.

At the end of the above algorithm, each $u \in V_H$ knows the list of all $v \in V_H$ such that $C(u,v) \neq \emptyset$. In order to let all vertices in $G$ learn the topology of $G_H$, it suffices to let all $u \in V_H$ broadcast this information.
This can be done using   \cref{lemma:bc}(2) with $x = \sqrt{n}$, which costs $\tilde{O}(n^{1.5})$ time and $\tilde{O}(\sqrt{n})$ energy.
\end{proof}

\begin{lemma}\label{lem:learn0}
In $\tilde{O}(n^{1.5})$ time and $\tilde{O}(\sqrt{n})$ energy, we can let all $u \in V$ learn $\mathcal{I}_0(u)$ w.h.p.
\end{lemma}
\begin{proof}
This follows from \cref{lem:S} and \cref{lem:pairs}.
\end{proof}

Next, we consider the task of learning  $\mathcal{I}_1(u)$ and  $\mathcal{I}_2(u)$.

\begin{lemma}\label{lem:learn12}
Suppose that  each $v \in V$  knows $\mathcal{I}_0(v)$.
Using $\tilde{O}(n^{1.5})$ time and  $\tilde{O}(\sqrt{n})$ energy, we can let all vertices $u \in V_H$ learn $\mathcal{I}_1(u)$ and  $\mathcal{I}_2(u)$ w.h.p.
\end{lemma}
\begin{proof}
 Consider any vertex $u \in V_H$. 
For each component $S \in C(u)$, we let $r_{S,u}$ be the smallest-ID vertex in the set $S \cap N(u)$.
For each $v \in V_H$ such that $F^\star(\{u,v\}) = u$, and for each  component $S \in C(u,v)$, we similarly let $r_{S,u}$ be the smallest-ID vertex in the set $S \cap N(u)$. 
As we will later see, $r_{S,u}$ will be the vertex in $S$ responsible for sending the graph topology $G^+[S]$ to $u$ in case $G^+[S]$ belongs to $\mathcal{I}_1(u)$ or  $\mathcal{I}_2(u)$. 

Recall that $\mathcal{I}_1(u)$ and  $\mathcal{I}_2(u)$ consist of the graph topology  $G^+[S']$ of some selected type-1 and type-2 components $S'$ such that $u$ belongs to $G^+[S']$.
We will present a generic approach that lets $u \in V_H$ learn one graph topology in $\mathcal{I}_1(u)$ and  $\mathcal{I}_2(u)$. As we will later see, the cost of learning one graph topology is $\poly \log n$ time and energy. If the graph topology to be learned is in $C(u)$, then only $u$ and the vertices $r_{S,u}$ for all $S \in C(u)$ need to participate in the algorithm for learning the graph topology. If the graph topology to be learned is in $C(u,v)$, then only $u$ and the vertices $r_{S,u}$ for all $S \in C(u,v)$ need to participate in the algorithm for learning the graph topology.
We only describe the algorithms that let $u \in V_H$ learn  $A_1[u]$ and $A_2[u]$. The algorithms for learning the remaining graph topologies are analogous.

\begin{description}
    \item[{Learning $A_1[u]$.}] Recall that $A_1[u]$ is a component $S' \in C(u)$ that maximizes $\ecc(u,S')$.
To learn  $A_1[u]$, we use $\maxsr$ with $\mathcal{S} =  \{r_{S, u} \, : \, S \in C(u)\}$ and $\mathcal{R} = \{u\}$.
The message $m_v$ of $v = r_{S, u}$ is  the graph topology of $G^+[S]$, and the key of $v = r_{S, u}$ is $k_v = \ecc(u,S)$.
Since each type-1 and type-2 component satisfies $|S| \leq \sqrt{n}$, the maximum  possible value of $\ecc(u,S)$ is $\sqrt{n}$, so
the size of the key space for $\maxsr$  is $K = \sqrt{n}$.

If $|C(u)| > 0$, then the message that $u$ receives from $\maxsr$ is the topology of $G^+[S']$, for a component $S' \in C(u)$ that attains the maximum value of $\ecc(u,S')$ among all components in $C(u)$,  so $u$ may set $A_1[u] = S'$.
If $|C(u)| = 0$, the vertex $u$ receives nothing from  $\maxsr$,  so $u$ may set  $A_1[u] = \emptyset$.
By \cref{lemma:min-sr}, the cost of  $\maxsr$ is $O(\log K \log \Delta \log n) = \poly \log n$.
    \item[{Learning $A_2[u]$.}] The procedure for learning $A_2[u]$ is almost exactly the same as that for $A_1[u]$, with only one difference.
Recall that $A_2[u]$ is a component $S' \in C(u) \setminus \{A_1[u]\}$ that maximizes $\ecc(u,S')$,  so we need to exclude the component $A_1[u]$ from participating. To do so, before we apply $\maxsr$, we  use one round to let $u$ send $\ID(r_{A_1[u],u})$ to all vertices $\{r_{S, u} \, : \, S \in C(u)\}$. This allows each $r_{S,u}$ to learn whether or not $S = A_1[u]$.
\end{description}

For each $u \in V_H$, the number of pairs $\{u,v\}$ such that $F^\star(\{u,v\}) = u$  is  $O(1)$, so the number of graph topologies needed to be learned in  $\mathcal{I}_1(u)$ and  $\mathcal{I}_2(u)$ by $u$ is $O(\sqrt{n})$. 
The total number of graph topologies needed to be learned, for all $u \in V_H$, is at most $|V_H| \cdot O(\sqrt{n}) = O(n)$. We fix an ordering of these learning tasks and solve them sequentially. For each of these tasks, we use the above generic approach to solve the task, so the time and energy cost for learning one graph topology is $\poly \log n$.
  Since there are $O(n)$ tasks, the overall time complexity is $O(n) \cdot \poly \log n = \tilde{O}(n)$. Each vertex participates in  $O(\sqrt{n})$ tasks,  so the overall energy complexity is $O(\sqrt{n}) \cdot \poly \log n = \tilde{O}(\sqrt{n})$.
\end{proof}

\begin{lemma}\label{lem:learnGstar}
   Using $\tilde{O}(n^{1.5})$ time and  $\tilde{O}(\sqrt{n})$ energy, we can let all vertices in $G$ learn the graph topology of $G^\star$ w.h.p.
\end{lemma}
\begin{proof}
The lemma follows from combining \cref{lem:Gstar,lem:learn0,lem:learn12}.
\end{proof}

Now we are ready to prove \cref{lem:diam-ub}.

 \thmdiammain*
\begin{proof}
    The theorem follows from combining \cref{lem:cal-diam,lem:learnGstar}.
\end{proof}

\section{Minimum cut \label{sect.cut}}

In this section, we apply the approach introduced in \cref{sect.planar} to show that (i) the exact global minimum cut size and (ii) an approximate $s$--$t$ minimum cut size of any bounded-genus graph can be computed in $\tilde{O}(\sqrt{n})$ energy. We also show energy lower bounds to complement these results.

\subsection{Global minimum cut \label{sect.cut-ub}}

In this section, we prove \cref{thm:mincut}. Both proofs follow the structure as the one in \cref{sect.planar}. That is, we still decompose the vertex set into $V_H$ and $V_L$, and we classify the connected components of $G[V_L]$ into three types. The only difference here is the information that we extract from type-1 and type-2 components. 

Given a cut $\mathcal{C} = (X, V \setminus X)$ of $G=(V,E)$, the two vertex sets $X \neq \emptyset$ and  $V \setminus X \neq \emptyset$ are called the two parts of $\mathcal{C}$, and the {\em cut edges} of $\mathcal{C}$ are defined as $\{ \{u,v\} \, : \, u \in X, v \in V \setminus X\}$. The \emph{size} of a cut $\mathcal{C}$, which we denote as $|\mathcal{C}|$, is defined as the number of cut edges of $\mathcal{C}$. A {\em minimum cut} of a graph is a cut $\mathcal{C}$ that minimizes $|\mathcal{C}|$ among all possible cuts. An {\em $s$--$t$ minimum cut} of a graph is a cut $\mathcal{C}$ the minimizes $|\mathcal{C}|$ among all possible cuts subject to the constraint that $s$ and $t$ belong to different parts. We consider the following definitions:
\begin{description}
\item [$c(S)$.] For any type-1 component $S$, let $c(S)$ be the minimum cut size of $G^+[S]$.
\item [$c'(S)$.] For any type-2 component $S \in C(u,v)$, let $c'(S)$ be the $u$--$v$ minimum  cut size of $G^+[S]$.
\item [$c''(S)$.]  For any type-2 component $S \in C(u,v)$, let $c''(S)$ be the minimum cut size of  $G^+[S]$ among all cuts such that both $u$ and $v$ are within the same part of the cut.
\end{description}

We make the following observations.

\begin{observation}\label{lem:cut-aux}
Let $\mathcal{C} = (X, V \setminus X)$ be any minimum cut of $G$. For any vertex $u \in V_H$, one of the following statements is true:
\begin{itemize}
\item One part of the cut contains all vertices in $\bigcup_{S \in C(u)} S \cup \{u\}$.
\item the size of the cut is  $\min_{S \in C(u)} c(S)$.
\end{itemize}
\end{observation}
\begin{proof}
Suppose that the  first statement is false. Then there exists a component $S' \in C(u)$ such that $S' \cup \{u\}$ intersects both parts of the cut, so $\mathcal{C}' = (X \cap (S' \cup \{u\}), (V \setminus X) \cap (S' \cup \{u\}))$ is a cut of $G^+[S']$. Therefore, $\min_{S \in C(u)} c(S) \leq c(S') \leq |\mathcal{C}'| \leq |\mathcal{C}|$. 
To prove that the second statement is true, we just need to show that $|\mathcal{C}| \leq \min_{S \in C(u)} c(S)$. This inequality follows from the observation that for any component $S \in C(u)$, any cut of $G^+[S]$ can be extended to a cut of $G$ of the same size by adding all vertices in $V \setminus (S \cup \{u\})$ to the part of the cut that contains $u$. 
\end{proof}

\begin{observation}\label{lem:cut-aux2}
Let $\mathcal{C} = (X, V \setminus X)$ be any minimum cut of $G$. For two distinct vertices $u,v \in V_H$, one of the following statements is true:
\begin{itemize}
\item One part of the cut contains all vertices in  $\bigcup_{S \in C(u,v)} S \cup \{u,v\}$.
\item The size of the cut is $\min_{S \in C(u,v)} c''(S)$.
\item $u$ and $v$ belong to different parts of the cut, and the number of cut edges that have at least one endpoint in $\bigcup_{S \in C(u,v)} S'$ is  $\sum_{S \in C(u,v)} c'(S)$.
\end{itemize}
\end{observation}
\begin{proof}
 Suppose that the first statement is false. We first focus on the case where $u$ and $v$ belong to the same part of the cut $\mathcal{C}$.
In this case, there exists a component $S' \in C(u,v)$ such that $S' \cup \{u,v\}$ intersects both parts of the cut, so $\mathcal{C}' = (X \cap (S' \cup \{u,v\}), (V \setminus X) \cap (S' \cup \{u,v\}))$ is a cut of $G^+[S]$ such that  $u$ and $v$ belong to the same part of the cut. Therefore, $\min_{S \in C(u,v)} c''(S) \leq c''(S') \leq |\mathcal{C}'| \leq |\mathcal{C}|$. Similar to the proof of \cref{lem:cut-aux}, we also have $|\mathcal{C}| \leq \min_{S \in C(u,v)} c''(S)$, as any cut of $G^+[S]$ such that  $u$ and $v$ belong to the same part of the cut can be extended to a cut of $G$ of the same size. Therefore, we must have $|\mathcal{C}| = \min_{S \in C(u,v)} c''(S)$, that is, the second statement is true.

For the rest of the proof,  we consider the case where  $u$ and $v$ belong to different parts of the cut $\mathcal{C}$. For each component $S \in C(u,v)$, we write $Z_{S}$ to denote the number of cut edges of $\mathcal{C}$ that have at least one endpoint in $S$. Then we must have $Z_{S} = c'(S)$, since otherwise $\mathcal{C}$ is not a minimum cut. Therefore, the number of cut edges that have at least one endpoint in $\bigcup_{S \in C(u,v)} S'$ is  $\sum_{S \in C(u,v)} c'(S)$, that is, the third statement is true.
\end{proof}

\paragraph{The graph $G^\diamond$.} Bounded-genus graphs have bounded arboricity.
The minimum degree of any graph of arboricity $\alpha$ is at most $2 \alpha - 1$.
The minimum cut size of any graph is at most the minimum degree of the graph.
Therefore, there is a constant $\lambda_0$ such that the minimum cut size of $G$ is at most $\lambda_0$.
We define the graph $G^\diamond$ as the result of applying the following operations to $G$:
 \begin{itemize}
     \item Remove all type-1 components.
     \item For each pair $\{u,v\}$ of distinct vertices in $V_H$ with $|C(u,v)| > 0$, replace $C(u,v)$ with $\min\{\lambda_0, \sum_{S \in C(u,v)} c'(S)\}$ multi-edges between $u$ and $v$.
 \end{itemize}

In the subsequent discussion, we say that a cut $\mathcal{C}$ of $G$ is \emph{good} if  it satisfies the following conditions:
\begin{itemize}
    \item For each vertex $u \in V_H$, one part of the cut contains all vertices in  $\bigcup_{S \in C(u)} S \cup \{u\}$.
    \item For any two distinct vertices $u, v \in V_H$, if   $u$ and $v$ belong to the same part of the cut, then this part of the cut contains all vertices in $\bigcup_{S \in C(u,v)} S \cup \{u,v\}$.
\end{itemize}

\begin{observation}\label{lem:cut-aux3}
If a minimum cut of $G$ is good, then both $G$  and $G^\diamond$ have the same minimum cut size.
\end{observation}
\begin{proof}
  This observation follows immediately from the construction of  $G^\diamond$. 
\end{proof}

Using \cref{lem:cut-aux,lem:cut-aux2,lem:cut-aux3}, we prove the following lemma.

\begin{lemma}\label{lem-mincut-aux}
The minimum cut size of $G$ is the minimum of the following numbers:
\begin{enumerate}
\item\label{ii1} The minimum value of $\min_{S \in C(u)} c(S)$ among all $u \in V_H$ such that $|C(u)|>0$.
\item\label{ii2} The minimum value of $\min_{S \in C(u,v)} c''(S)$ among all $u,v \in V_H$ such that  $|C(u,v)|>0$.
\item\label{ii3} The minimum cut size of $G^\diamond$.
\end{enumerate}
\end{lemma}
\begin{proof}
For each $S \in C(u)$, there exists a cut of $G^+[S]$ of size $c(S)$, and such a cut can be extended to a cut of $G$ of the same size by adding all vertices in $V \setminus (S \cup \{u\})$ to the part of the cut that contains $u$. 
Similarly, for each $S \in C(u,v)$, there exists a cut of $G^+[S]$ of size $c''(S)$ where both $u$ and $v$ belong to the same part, and such a cut can be extended to a cut of $G$ of the same size by adding all vertices in $V \setminus (S \cup \{u,v\})$ to the part of the cut that contains $u$ and $v$. Therefore,   the minimum cut size of $G$  is at most the minimum value of \cref{ii1,ii2}. 

By \cref{lem-mincut-aux}, we infer that the minimum cut size of $G$  is also at most the value of \cref{ii3}, so now we know that the minimum cut size of $G$  is at most the minimum value of \cref{ii1,ii2,ii3}. To finish the proof, we will show that  the minimum cut size of $G$  is at least the minimum value of \cref{ii1,ii2,ii3}. To do so, we assume that the minimum cut size of $G$  is smaller than the minimum value of \cref{ii1,ii2}, and then our goal is to show that the minimum cut size of $G$ is at least the value of \cref{ii3}. By \cref{lem:cut-aux,lem:cut-aux2}, such an assumption implies that any minimum cut of $G$ is good, so its size equals the minimum cut size of $G^\diamond$ by \cref{lem:cut-aux3}.
%Let $\mathcal{C} = (X, V \setminus X)$ be any minimum cut of $G$. 
%We construct a cut $\mathcal{C}^\diamond = (X^\diamond, V^\diamond \setminus X^\diamond)$ of $G^\diamond=(V^\diamond, E^\diamond)$ as follows.
%\begin{itemize}
%    \item For each vertex $v$ in $G$ that does not belong to a type-1 or a type-2 component, add $v$ to $X^\diamond$ if $v \in X$.
%\end{itemize}
%We claim that the size of $\mathcal{C}$ equals the size of $\mathcal{C}^\diamond$.
%By \cref{lem:cut-aux}, $\mathcal{C}$ does not have any cut edge in $G^+[S]$ for any type-1 component. Similarly, by \cref{lem:cut-aux2}, $\mathcal{C}$ does not have any cut edge in $G^+[S]$ for any type-2 component $S \in C(u,v)$ such that both $u$ and $v$ are within the same part of cut $\mathcal{C}$.
\end{proof}
  
\paragraph{Information.} For each vertex $u \in V$, we define $\mathcal{I}^\diamond_0(u)$, $\mathcal{I}^\diamond_1(u)$, and $\mathcal{I}^\diamond_2(u)$ as follows. 
\begin{itemize}
    \item $\mathcal{I}^\diamond_0(u)$ is the same as the basic information $\mathcal{I}_0(u)$ defined in \cref{sect.planar}.
    \item $\mathcal{I}^\diamond_1(u)$ contains the number $\min_{S \in C(u)} c(S)$.
    \item $\mathcal{I}^\diamond_2(u)$ contains the two numbers $\min_{S \in C(u,v)} c''(S)$ and $\min\{\lambda_0, \sum_{S \in C(u,v)} c'(S)\}$, for all pairs $\{u,v\} \in E(G_H)$ such that $F^\star(\{u,v\}) = u$.
\end{itemize}

Note that $\mathcal{I}_1(u)$ and $\mathcal{I}_2(u)$ contain nothing if $u \in V_L$.

\thmcutmain*
 
\begin{proof}
As $\mathcal{I}^\diamond_0(u) = \mathcal{I}_0(u)$, we may use the algorithm of \cref{lem:learn0} to let all vertices $u \in V$ learn the information  $\mathcal{I}^\diamond_0(u)$ using $\tilde{O}(n^{1.5})$ time and $\tilde{O}(\sqrt{n})$ energy.

The algorithm of \cref{lem:learn12} can be modified to allow 
all vertices $u \in V_H$ learn the information  $\mathcal{I}^\diamond_1(u)$ and $\mathcal{I}^\diamond_2(u)$. 
Specifically, the number $\min_{S \in C(u)} c(S)$ can be learned by the same algorithm for learning  $A_1[u]$ described in the proof of \cref{lem:learn0} by replacing 
$\maxsr$ with $\minsr$ and letting $v = r_{S,u}$ use the key $k_v = c(S)$.
The algorithm for learning $\min_{S \in C(u,v)} c''(S)$ is similar. 

For each pair $\{u,v\} \in E(G_H)$ such that $F^\star(\{u,v\}) = u$, to let $u$ learn $\min\{\lambda_0, \sum_{S \in C(u,v)} c'(S)\}$, we use $\apxsr$ with the following parameters:
\begin{itemize}
\item $\mathcal{S} =  \{r_{S, u} \, : \, S \in C(u,v)\}$, where $r_{S,u}$ is the smallest-ID vertex in the set $S \cap N(u)$.
\item $\mathcal{R} = \{u\}$.
    \item $\epsilon = 1/(2\lambda_0 + 1)$.
    \item $W = \lambda_0$.
    \item For each $S \in C(u,v)$, the message $m_v$ of the representative $v = r_{S, u}$ of $S$ is $\min\{\lambda_0, c'(S)\}$.
\end{itemize}
After the algorithm of $\apxsr$, $u$ learns a $(1\pm \epsilon)$-approximation of \[\sum_{v \in N^+(u) \cap \mathcal{S}} m_v = \sum_{S \in C(u,v)} \min\{\lambda_0, c'(S)\}.\] We claim that this allows $u$ to calculate  $\min\{\lambda_0, \sum_{S \in C(u,v)} c'(S)\}$ precisely. To prove this claim, we break the analysis into two cases. Let $x$ be the approximation of  $\sum_{S \in C(u,v)} \min\{\lambda_0, c'(S)\}$ computed by $\apxsr$.

 If  $\min\{\lambda_0, \sum_{S \in C(u,v)} c'(S)\} = \lambda_0$, then \[\sum_{v \in N^+(u) \cap \mathcal{S}} m_v =  \sum_{S \in C(u,v)} \min\{\lambda_0, c'(S)\} \geq \lambda_0,\] which implies \[x \geq (1-\epsilon)\lambda_0 > \lambda_0 - 1/2.\]
 
 If  $\min\{\lambda_0, \sum_{S \in C(u,v)} c'(S)\} = \sum_{S \in C(u,v)} c'(S)$, then \[\sum_{v \in N^+(u) \cap \mathcal{S}} m_v =  \sum_{S \in C(u,v)} \min\{\lambda_0, c'(S)\} = \sum_{S \in C(u,v)} c'(S),\] which implies 
\begin{align*}
x &\in \left[(1-\epsilon)\sum_{S \in C(u,v)} c'(S),  (1+\epsilon)\sum_{S \in C(u,v)} c'(S)\right]\\
&\subseteq \left(\left(\sum_{S \in C(u,v)} c'(S)\right)- \frac{1}{2},  \left(\sum_{S \in C(u,v)} c'(S)\right) +\frac{1}{2}\right).    
\end{align*} 
 Therefore, $u$ can  calculate  $\min\{\lambda_0, \sum_{S \in C(u,v)} c'(S)\}$ precisely from $x$.
By \cref{lemma:apxsr}, the cost for $u$ to calculate  $\min\{\lambda_0, \sum_{S \in C(u,v)} c'(S)\}$ via $\apxsr$ is $\poly \log n$ time. 

For each $u \in V_H$, the number of pairs $\{u,v\}$ such that $F^\star(\{u,v\}) = u$  is  $O(1)$, so the number of parameters needed to be learned in  $\mathcal{I}^\diamond_1(u)$ and  $\mathcal{I}^\diamond_2(u)$ by $u$ is $O(1)$. 
The total number of parameters needed to be learned, for all $u \in V_H$, is at most $|V_H| \cdot O(1) = O(\sqrt{n})$. We fix any ordering of these learning tasks and solve them sequentially. The time and energy cost for learning one parameter is $\poly \log n$.
  Since there are $O(\sqrt{n})$ tasks, the overall time complexity for learning $\mathcal{I}^\diamond_1(u)$ and  $\mathcal{I}^\diamond_2(u)$ for all $u \in V_H$ is $O(\sqrt{n}) \cdot \poly \log n = \tilde{O}(\sqrt{n})$.

In view of \cref{lem-mincut-aux}, the minimum cut size of $G$ can be calculated from the following information: (i) $\mathcal{I}^\diamond_1(u)$ and $\mathcal{I}^\diamond_2(u)$ for all $u \in V_H$, (ii) the
 topology of $G^+[S]$ for each type-3 component $S$, and (iii) the topology of the subgraph induced
by $V_H$. By replacing $\mathcal{I}_1(u)$ and $\mathcal{I}_2(u)$ with $\mathcal{I}^\diamond_1(u)$ and $\mathcal{I}^\diamond_2(u)$ in the description of the algorithm of \cref{lem:Gstar}, we obtain an algorithm that
 lets all vertices learn this information using $\tilde{O}(n^{1.5})$ time and $\tilde{O}(\sqrt{n})$ energy.
\end{proof}

\subsection{Approximate \texorpdfstring{$s$--$t$}{st} minimum cut \label{sect.cut-ub2}}

In this section, we prove \cref{thm:stmincut}. The proof of \cref{thm:stmincut} is similar to that of  \cref{thm:mincut}. The main difference for the setting of $s$--$t$ minimum cut is that if $s$ or $t$ happens to be within a type-1 or a type-2 component $S$, then we additionally need to learn the  topology of $G^+[S]$.
Any type-1 component that does not contain $s$ or $t$ is irrelevant to the $s$--$t$ minimum cut size.

In the subsequent discussion, we fix $s$ and $t$ to be any two distinct vertices of $G$. for each $x \in \{s,t\}$, let $S_x$ be the type-1 or type-2 component containing $x$. In case $x$ is not contained in any type-1 or type-2 component, we let $S_x = \emptyset$.
We define $G^\bullet$ as the result of applying the following operations to $G$. 
\begin{itemize}
    \item Remove all type-1 components, except for $S_s$ and $S_t$.
    \item For each pair $\{u,v\}$ of distinct vertices in $V_H$ with $|C(u,v) \setminus \{S_s, S_t\}| > 0$, replace all components in $C(u,v) \setminus \{S_s, S_t\}$ with  $\sum_{S \in C(u,v) \setminus \{S_s, S_t\}} c'(S)$ multi-edges between $u$ and $v$.
\end{itemize}

Similar to \cref{lem:cut-aux,lem:cut-aux2}, we have the following observation.

\begin{observation}\label{obs2}
Both $G$ and $G^\bullet$ have the same minimum $s$--$t$ cut size.
\end{observation}
\begin{proof}
Fix $\mathcal{C} = (X, V \setminus X)$ to be any minimum $s$--$t$ cut of $G$, where $s \in X$ and $t \in V\setminus X$. To show that both $G$ and $G^\bullet$ have the same minimum $s$--$t$ cut size, it suffices to show the following two statements:
\begin{itemize}
    \item For each type-1 component $S$ that is not $S_s$ and $S_t$, we must have either $S \subseteq X$ or $S \subseteq V\setminus X$.
    \item For each pair $\{u,v\}$ of distinct vertices in $V_H$ with $|C(u,v) \setminus \{S_s, S_t\}| > 0$, if $u$ and $v$ belong to different parts of cut $\mathcal{C}$, then the number of cut edges of $\mathcal{C}$ with at least one endpoint in $\bigcup_{S \in C(u,v) \setminus \{S_s, S_t\}} S$ equals $\sum_{S \in C(u,v) \setminus \{S_s, S_t\}} c'(S)$.
\end{itemize}

The first statement follows from the observation that for each $u \in V_H$, all vertices in $\bigcup_{S \in C(u) \setminus \{S_s, S_t\}} S$ must belong to the part of cut $\mathcal{C}$ that $u$ belongs to, since otherwise $\mathcal{C}$ is not a minimum $s$--$t$ cut, as moving all vertices in $\bigcup_{S \in C(u) \setminus \{S_s, S_t\}} S$ to the part of cut that $u$ belongs to reduces the number of cut edges.

To show the second statement, consider a pair $\{u,v\}$ of distinct vertices in $V_H$ with $|C(u,v) \setminus \{S_s, S_t\}| > 0$ such that $u$ and $v$ belong to different parts of cut $\mathcal{C}$. Similar to the proof of \cref{lem:cut-aux2}, for each component $S \in C(u,v) \setminus \{S_s, S_t\}$, we write $Z_{S}$ to denote the number of cut edges of $\mathcal{C}$ that have at least one endpoint in $S$. Then we must have $Z_{S} = c'(S)$, since otherwise $\mathcal{C}$ is not a minimum cut. Therefore, the number of cut edges of $\mathcal{C}$ that have at least one endpoint in $\bigcup_{S \in C(u,v) \setminus \{S_s, S_t\}} S'$ is  $\sum_{S \in C(u,v) \setminus \{S_s, S_t\}} c'(S)$.
\end{proof}

We are ready to prove \cref{thm:stmincut}.

\thmcutapxmain*

\begin{proof}
The proof is very similar to the proof of \cref{thm:mincut}, so here we only describe the differences.
Let $\tilde{G}^\bullet$ be any graph such that for each pair of vertices $\{u,v\}$, the number of multi-edges in $\tilde{G}^\bullet$ is within a $(1 \pm \epsilon)$ factor
 of the number of multi-edges in $G^\bullet$.
By \cref{obs2}, the minimum $s$--$t$ cut size in  $\tilde{G}^\bullet$ is a $(1 \pm \epsilon)$-approximation of the minimum $s$--$t$ cut size of $G$.
Therefore, the task of computing the minimum $s$--$t$ cut size of $G$ is reduced to computing such a graph $\tilde{G}^\bullet$.

For each $u \in V_H$, we let  $\mathcal{I}_2^{\bullet}(u)$ contain the number $\sum_{S \in C(u,v) \setminus \{S_s, S_t\}} c'(S)$ for all pairs $\{u,v\} \in E(G_H)$ with $F^\star(\{u,v\}) = u$.
The same algorithm for learning $\mathcal{I}_2^{\diamond}(u)$ presented in the proof of \cref{thm:mincut} can be applied here to let all  $u \in V_H$ learn $\mathcal{I}_2^{\bullet}(u)$. Specifically, for each pair $\{u,v\} \in E(G_H)$ such that $F^\star(\{u,v\}) = u$, to let $u$ learn $\sum_{S \in C(u,v) \setminus \{S_s, S_t\}} c'(S)$, we use $\apxsr$ with $\epsilon$ and the following parameters:
\begin{itemize}
\item $\mathcal{S} =  \{r_{S, u} \, : \, S \in C(u,v) \setminus \{S_s, S_t\}\}$, where $r_{S,u}$ is the smallest-ID vertex in the set $S \cap N(u)$.
\item $\mathcal{R} = \{u\}$.
    \item $W = \binom{n}{2}$ is an upper bound on $|E| \geq c'(S)$ for any $S$.
    \item For each $S \in C(u,v) \setminus \{S_s, S_t\}$, the message $m_v$ of the representative $v = r_{S, u}$ of $S$ is $c'(S)$.
\end{itemize}
By \cref{lemma:apxsr}, the round complexity of $\apxsr$ is $\poly(\log n, 1/\epsilon)$. For each $u \in V_H$, the number of pairs $\{u,v\}$ such that $F^\star(\{u,v\}) = u$  is  $O(1)$, so the number of parameters needed to be learned in  $\mathcal{I}_2^\bullet(u)$ by $u$ is $O(1)$. 
The total number of parameters needed to be learned, across all $u \in V_H$, is at most $|V_H| \cdot O(1) = O(\sqrt{n})$. 
  Since there are $O(\sqrt{n})$ learning tasks and each vertex participates in $O(1)$ of them, the overall cost for learning   $\mathcal{I}_2^\bullet(u)$ for all $u \in V_H$ is $\tilde{O}(\sqrt{n}) \cdot \epsilon^{-O(1)}$ time and $\poly(\log n, 1/\epsilon)$ energy.

By \cref{obs2}, a $(1\pm \epsilon)$-approximation of the minimum $s$--$t$ cut size of $G$ can be calculated from the following information: 
(i)  $\mathcal{I}_2^\bullet(u)$ for all $u \in V_H$, (ii) the
 topology of $G^+[S]$ for $S = S_s$, $S = S_t$, and each type-3 component $S$, and (iii) the topology of the subgraph induced
by $V_H$, as they allow us to obtain the desired graph $\tilde{G}^\bullet$. Same as the proof of \cref{thm:mincut}, we may let all vertices learn this information  using $\tilde{O}(n^{1.5})$ time and $\tilde{O}(\sqrt{n})$ energy.

Hence there is an algorithm that computes a $(1 \pm \epsilon)$-approximate $s$--$t$ minimum cut size in $\tilde{O}(n^{1.5}) + \tilde{O}(\sqrt{n}) \cdot \epsilon^{-O(1)}$ time and $\tilde{O}(\sqrt{n} + \epsilon^{-O(1)})$ energy w.h.p.
\end{proof}

\subsection{Lower bounds \label{sect.cut-lb}}

In this section, we prove the two lower bounds: \cref{lem:learn-deg,lem:learn-deg2}.

\thmLBa*

\begin{proof}
Suppose that there is a randomized algorithm $\mathcal{A}$ that computes the exact $s$--$t$ minimum cut size of any planar bipartite graph with high probability and using $o(n)$ energy.

Let $G$ be a complete bipartite graph $K_{2,\Delta}$ with the bipartition $\{s, t\} \cup \{v_1, \ldots, v_{\Delta}\}$.
Set $X = \Delta/5$.
We select $\Delta$ to be sufficiently large so that it is guaranteed that 
both $s$ and $t$ use at most $X$ unit of energy in an execution of $\mathcal{A}$ on $G$.

Let $G'$ be the result of removing $v_{\Delta}$ from $G$.
The size of a $s$--$t$  minimum cut of $G$ is $\Delta$, and the size of a $s$--$t$  minimum cut of $G'$ is $\Delta - 1$.
Therefore, $\mathcal{A}$ allows $s$ to correctly distinguish between $G$ and $G'$ with high probability.

Consider an execution of $\mathcal{A}$ on $G$. Let $S$ be the subset of $\{v_1, \ldots, v_{\Delta}\}$ such that $v_i \in S$ if there is a time slot $\tau$ where (i) $v_i$ transmits, (ii) the number of vertices in $\{v_1, \ldots, v_{\Delta}\}$ that transmit is at most 2, and (iii) at least one of $s$ and $t$ listens.

We claim that $|S| \leq 4X = 4\Delta/5$. Let $T$ be the set of all time slots $\tau$ such that the above conditions (i), (ii), and (iii) hold for at least one $v_i \in \{v_1, \ldots, v_{\Delta}\}$. In view of condition (ii), we must have $|T| \geq |S|/2$. In view of condition (iii),  if $\tau \in T$, then at least one of $s$ and $t$ must listen at time $\tau$, so the energy cost of one of $s$ and $t$ must be at least $|T| / 2 \geq |S|/ 4$, which implies $X \geq |S| / 4$.

Let $\mathcal{E}$ be the event that $v_{\Delta} \notin S$ in an execution of $\mathcal{A}$ on $G$.
Whether or not $\mathcal{E}$ occurs depends only on the local randomness stored in the vertices $\{s, t\}$ and $\{v_1, \ldots, v_{\Delta}\}$. Since $|S| \leq 4\Delta/5$, at least $1/5$ fraction of the vertices in 
$\{v_1, \ldots, v_{\Delta}\}$ are not in $S$.
Since the probability that $v_{i} \notin S$ is identical for all $v_i \in \{v_1, \ldots, v_{\Delta}\}$, we have $\Prob[\mathcal{E}] \geq 1/5$.

Consider the following scenario.
All vertices in $\{s, t\}$ and $\{v_1, \ldots, v_{\Delta}\}$ have decided their random bits in advance.
With probability $1/2$, we run $\mathcal{A}$ on $G$. With probability $1/2$, we run $\mathcal{A}$ on $G'$.
If $\mathcal{E}$ occurs, then the execution of  $\mathcal{A}$ on both $G$ and $G'$ is completely identical from the point of view of each vertex, except for $v_{\Delta}$. Therefore, conditioning on  event $\mathcal{E}$, the probability that vertex $s$  correctly decides whether the underlying graph is $G$ or $G'$ is at most $1/2$, as $s$ can only guess randomly.

Since $\Prob[\mathcal{E}] \geq  1/5$, the probability that vertex $s$ fails to correctly decide whether the underlying graph is $G$ or $G'$ is at least $(1/2) \cdot (1/5) = 1/10$,  so $s$ fails to correctly calculate the  $s$--$t$ minimum cut with probability at least  $1/10$ in the above scenario. This contradicts the assumption that $\mathcal{A}$ is able to compute the $s$--$t$ minimum cut with high probability.
\end{proof}

The lower bound of \cref{lem:learn-deg} can be expressed in terms of the maximum degree $\Delta$. For  graphs with maximum degree $\Delta$, the  proof of \cref{lem:learn-deg} shows an $\Omega(\Delta)$ energy lower bound.

\thmLBb*

\begin{proof}
Consider the case where the underlying graph is $K_n$ with probability $1/2$, and is $K_n - e$  with probability $1/2$, where the edge $e$ is chosen uniformly at random from the set of all edges in $K_n$.
Let $\mathcal{A}$ be any randomized algorithm that computes the size of a minimum cut exactly with high probability. Observe that the  size of a  minimum cut of  $K_n$ is $n-1$ and the  size of a  minimum cut of $K_n - e$ is $n-2$,  so  $\mathcal{A}$ is able to  distinguish between $K_n$  and $K_n - e$  with high probability. It was shown in~\cite{Chang20bfs} that any  algorithm that distinguishes between $K_n$  and $K_n - e$ with success probability at least $3/4$ necessarily has energy cost $\Omega(n)$ in both $\cd$ and $\nocd$, so the randomized energy complexity of $\mathcal{A}$ is $\Omega(n)$.
\end{proof}

\section*{Acknowledgments} The work was partly done when the author was a student at the University of Michigan. The author would like to thank his advisor, Seth Pettie, for valuable discussions on this research topic and helpful comments on earlier drafts of this paper.

\bibliographystyle{abbrv}
\bibliography{references}

\newpage
\appendix

\section*{\LARGE Appendix}

\section{Algorithms for communication between two sets of vertices\label{sect.sr}}
In this section, we present our algorithms for $\sr$ and its variants. 
Recall that $\sr$ requires that 
each vertex $v \in \mathcal{R}$ with $N^+(v) \cap \mathcal{S} \neq \emptyset$ receives a message
$m_u$ from {\em at least one} vertex $u \in N^+(v) \cap \mathcal{S}$ w.h.p.

\begin{lemma}[\cite{bar1992time}]\label{lemma:sr}
$\sr$ can be solved in time $O(\log \Delta \log n)$ and energy $O(\log \Delta \log n)$.
\end{lemma}
\begin{proof}
By the definition of $\sr$, each vertex $v \in \mathcal{S} \cap \mathcal{R}$ is not required to receive any message from other vertices, as we already have $v \in N^+(v) \cap \mathcal{S}$.
Therefore, in the subsequent discussion, we assume that $\mathcal{S} \cap \mathcal{R} = \emptyset$.

The task $\sr$ with $\mathcal{S} \cap \mathcal{R} = \emptyset$ can be solved using the well-known \emph{decay} algorithm of~\cite{bar1992time}, which repeats the following routine for $C \log n$ times: For $i = 1$ to $\log \Delta$, let each vertex $u \in \mathcal{S}$ transmit with probability $2^{-i}$. Each $v \in \mathcal{R}$ is always listening throughout the procedure. Here $C > 0$ is some large enough constant to be determined.

Consider a vertex $v \in \mathcal{R}$ such that $N(v) \cap \mathcal{S} \neq \emptyset$.
Let $i^\star$ be the largest integer $i$ such that $2^{i} \leq 2|N(v) \cap \mathcal{S}|$.
Consider a time slot $t$ where each vertex $u \in \mathcal{S}$ transmits with probability $2^{-i^\star}$.
For notational simplicity, we write $n' = |N(v) \cap \mathcal{S}|$ and $p' = 2^{-i^\star}$.
Our choice of $i^\star$  implies that $1/n' \geq p' \geq 1/(2n')$.
The probability of the event that exactly one vertex in the set $N(v) \cap \mathcal{S}$ transmits equals $n' p' (1 - p')^{n'-1} \geq 1/(2e)$. The calculation follows from the inequalities $n' p' \geq 1/2$ and $(1 - p')^{n'-1} \geq (1 - 1/n')^{n'-1} \geq 1/e$.

If the above event occurs, then $v$ successfully receives a message $m_u$ from a vertex $u \in N(v) \cap \mathcal{S}$.
The probability that $v$ does not receive any message from vertices in $N(v) \cap \mathcal{S}$ throughout the entire algorithm is at most $(1 - 1/(2e))^{C \log n} = n^{-\Omega(C)}$. By setting $C$ to be a large enough constant, the algorithm successfully solves $\sr$ w.h.p., and the time and energy complexities of the algorithm are $O(\log \Delta \log n)$.
\end{proof}

Recall that the goal of $\allsr$ is to let each vertex  $u \in \mathcal{S} \cap N^+(v)$ deliver a message $m_u$ to $v \in \mathcal{R}$, for each $v \in \mathcal{R}$.

\begin{lemma}\label{lemma:allsr}
$\allsr$ can be solved in time $O(\Delta' \log n)$ and energy $O(\Delta' \log n)$, where $\Delta'$ is an upper bound on $|\mathcal{S} \cap N(v)|$, for each $v \in \mathcal{R}$.
\end{lemma}
\begin{proof}
Consider the algorithm which repeats the following routine for $C \cdot \Delta' \log n$ rounds, for some sufficiently large constant $C > 0$. In each round, each vertex $u \in \mathcal{S}$ sends $m_u$ with probability $1 / \Delta'$. For each $u \in \mathcal{R}$, if $u$ does not send in this round, then $u$ listens.

Let $e = \{u,v\}$ be any edge with $u \in \mathcal{S}$ and $v \in \mathcal{R}$. In one round of the above algorithm,  $u$ successfully sends a message to $v$ if (i) all vertices in $\{v\} \cup (\mathcal{S} \cap N(v)) \setminus \{u\}$ do not send, and (ii) $u$ sends.
Therefore, the probability that $u$ successfully sends a message to $v$ is
\[
(1 - 1/\Delta')^{|\mathcal{S} \cap N(v)|-1} \cdot (1/\Delta') \geq  (1 - 1/\Delta')^{\Delta'-1} \cdot (1/\Delta') \geq 1/(e\Delta')
\]
The probability that $u$ does not successfully send a message to $v$ throughout all $C \cdot \Delta' \log n$ rounds is at most $(1 - 1/(e\Delta'))^{C \cdot \Delta' \log n} = n^{-\Omega(C)}$. Selecting a large enough constant $C$, by a union bound for all $u \in \mathcal{S} \cap N(v)$ and all $v \in \mathcal{R}$, we conclude that the algorithm solves $\allsr$ w.h.p. The time and energy complexities are $O(\Delta' \log n)$.
\end{proof}

 Recall that the task $\multisr$ requires that each vertex  $v \in \mathcal{R}$ receive all distinct messages in $\bigcup_{u \in N^+(v) \cap \mathcal{S}} \mathcal{M}_u$, where is the $\mathcal{M}_u$ is the set of messages hold by $u$.

\begin{lemma}\label{lemma:multisr}
$\multisr$ can be solved in time $O(M \log \Delta \log^2 n)$ and energy $O(M  \log \Delta \log^2 n)$, where $M$ is an upper bound on the number of distinct messages in $\bigcup_{u \in N^+(v) \cap \mathcal{S}}$, for each $v \in \mathcal{R}$.
\end{lemma}
\begin{proof}
Consider the algorithm which repeatedly runs  $\sr$  for $C \cdot M \log n$ times, where in each iteration, the sets $(\mathcal{S}',\mathcal{R'})$ for $\sr$ are chosen randomly as follows.
We select $\mathcal{R'}$ as a random subset of  $\mathcal{R}$ such that each $v \in \mathcal{R}$ joins $\mathcal{R'}$ with probability $1/2$.
We select $\mathcal{S}'$ as a random subset of $\mathcal{S} \setminus \mathcal{R'}$ such that for each message $m$, all vertices in  $\mathcal{S} \setminus \mathcal{R'}$  that hold $m$ join $\mathcal{S}'$  with probability $1/M$, using the shared randomness associated with the message $m$.

Due to the shared randomness, if $u \in \mathcal{S} \setminus \mathcal{R'}$ joins $\mathcal{S}'$ due to message $m$, then all vertices in $\mathcal{S} \setminus \mathcal{R'}$ holding the same message $m$ also joins $\mathcal{S}'$. Note that a vertex $u \in \mathcal{S} \setminus \mathcal{R'}$ might hold more than one message in that $|\mathcal{M}_u| > 1$. The probability that $u \in \mathcal{S} \setminus \mathcal{R'}$ joins $\mathcal{S}'$ equals $\Prob[\binomial(|\mathcal{M}_u|, 1/M) \geq 1]$, because each message $m \in \mathcal{M}_u$ lets $u$ join   $\mathcal{S}'$ with probability $1/M$ independently.

To analyze the algorithm, we focus on one vertex $v \in \mathcal{R}$ in one iteration of the above algorithm.
Consider any message $m \in \bigcup_{u \in N(v) \cap \mathcal{S}} \mathcal{M}_u \setminus \mathcal{M}_v$.
Observe that $v$ receives $m$ if the following three  events $\mathcal{E}_1$, $\mathcal{E}_2$, and $\mathcal{E}_3$ occur:
\begin{itemize}
    \item $\mathcal{E}_1$ is the event that $v$ joins $\mathcal{R}'$.
    \item $\mathcal{E}_2$ is the event that at least one vertex $u \in N(v) \cap \mathcal{S}$ with $m \in \mathcal{M}_u$ does not join $\mathcal{R}'$.
    \item $\mathcal{E}_3$ is the event that the subset of vertices of $N(v) \cap \mathcal{S} \setminus \mathcal{R'}$  joining $\mathcal{S}'$ is exactly the set of all vertices $u \in N(v) \cap \mathcal{S} \setminus \mathcal{R'}$ with $m \in \mathcal{M}_u$.
\end{itemize}
If $\mathcal{E}_1$, $\mathcal{E}_2$, and $\mathcal{E}_3$ occur, then $v \in \mathcal{R}'$, $N(v) \cap \mathcal{S}' \neq \emptyset$, and all vertices $u \in N(v) \cap \mathcal{S}'$ satisfy $m \in \mathcal{M}_u$. Therefore, conditioning on $\mathcal{E}_1$, $\mathcal{E}_2$, and $\mathcal{E}_3$, $\sr$ in this iteration allows $v$ to receive message $m$.

The way $\mathcal{R}'$ is selected implies that  $\Prob[\mathcal{E}_1] = 1/2$ and $\Prob[\mathcal{E}_2] \geq 1/2$. Observe that $\mathcal{E}_1$ and $\mathcal{E}_2$ are independent events. The way $\mathcal{S}'$ is selected implies that  $\Prob[\mathcal{E}_3 | \mathcal{E}_1 \cap \mathcal{E}_2] \geq \Prob[\binomial(M, 1/M)=1] = (1/M) \cdot (1 - 1/M)^{M-1} \geq 1/(eM)$. Therefore, the probability that $v$ receives $m$ in this iteration is at least $1/(4eM)$.

The probability that $v$ does not receive $m$ in all iterations is at most $(1 - 1/(4eM))^{C \cdot M \log n} = n^{-\Omega(C)}$. Selecting a large enough constant $C$, by a union bound for all $v \in \mathcal{R}$ and all $m \in \bigcup_{u \in N(v) \cap \mathcal{S}} \mathcal{M}_u \setminus \mathcal{M}_v$,  we conclude that the algorithm solves $\allsr$ w.h.p. The time and energy complexities are $O(M  \log \Delta \log^2 n)$, as the number of iterations is $O(M \log n)$ and the time complexity of each iteration is $O(\log \Delta \log n)$ by \cref{lemma:sr}.
\end{proof}

Consider the setting where the message $m_u$ sent from each vertex $u \in \mathcal{S}$ contains a key $k_u$ from the key space $[K] = \{1, 2, \ldots, K\}$. Recall that  $\minsr$ requires that each vertex $v \in \mathcal{R}$ with $N^+(v) \cap \mathcal{S} \neq \emptyset$ receives a message $m_u$  from a vertex $u \in N^+(v) \cap \mathcal{S}$ such that $k_u = \min_{u' \in N^+(v) \cap \mathcal{S}} k_{u'}$.

\begin{lemma}\label{lemma:min-sr}
Both $\minsr$ and $\maxsr$  can be solved in time $O(K \log \Delta \log n)$ and energy $O(\log K \log \Delta \log n)$.
For the special case of $\mathcal{S} \cap \mathcal{R} = \emptyset$ and $|\mathcal{R} \cap N(u)| \leq 1$ for each $u \in \mathcal{S}$, the time complexity  can be improved to  $O(\log K \log \Delta \log n)$.
\end{lemma}
\begin{proof}
We only prove the lemma for $\minsr$, as the proof for $\maxsr$ is the same.
The proof presented here is analogous to the analysis of a deterministic version of $\sr$ in~\cite{ChangDHHLP18}.
Observe that we can do $\sr$ once to let each $v \in \mathcal{R}$ test whether or not $N^+(v) \cap \mathcal{S} \neq \emptyset$.
If a vertex $v \in \mathcal{R}$  knows that $N^+(v) \cap \mathcal{S} = \emptyset$, then $v$ may remove itself from $\mathcal{R}$.
Thus, in the subsequent discussion, we assume $N^+(v) \cap \mathcal{S} \neq \emptyset$ for each $v \in \mathcal{R}$.

Let $v \in \mathcal{R}$, and we define $f_v = \min_{u \in N^+(v) \cap \mathcal{S}} k_u$.
The high-level idea of the algorithm is to conduct a binary search to determine all $\log K$ bits of the binary representation of $f_v$.

\paragraph{General case.}
Suppose at some moment each vertex $v \in \mathcal{R}$ already knows the first $x$ bits of $f_v$.
The following procedure allows each $v \in \mathcal{R}$ to learn the $(x+1)$th bit of $f_v$.
For each $(x+1)$-bit binary string $s$, we do $\sr$ with the following choices of $(\mathcal{S}',\mathcal{R}')$:
\begin{itemize}
    \item $\mathcal{S}'$ is the set of vertices $u \in \mathcal{S}$ such that the first $x+1$ bits of $k_u$ equal $s$.
    \item $\mathcal{R}'$ is the set of vertices $v \in \mathcal{R}$ such that the first $x$ bits of $f_v$ equal the first $x$ bits of $s$.
\end{itemize}
In this procedure, we perform $2^{x+1}$ times of $\sr$ in total, but each vertex only participates in at most three of them, as each vertex  joins $\mathcal{S}'$ at most once and joins $\mathcal{R}'$ at most twice.
Thus, the procedure costs $O(2^x \log \Delta \log n)$ time and $O(\log \Delta \log n)$ energy, by \cref{lemma:sr}. For each $v \in \mathcal{R}$, the messages that $v$ receive during the procedure allows $v$ to determine the $(x+1)$th bit of $f_v$. 

We will run the above procedure for $\log K$ iterations from $x=0$ to $x = \log K - 1$. Observe that in the last iteration, each vertex $v \in \mathcal{R}$ is guaranteed to receive a message $m_u$ from a vertex $u \in N^+(v) \cap \mathcal{S}$ such that $k_u = f_v = \min_{w \in N^+(v) \cap \mathcal{S}} k_{w}$, so this algorithm allows us to solve $\minsr$.
The overall time complexity of the algorithm  is \[\sum_{x=0}^{\log K - 1} O(2^x \log \Delta \log n) = O(K \log \Delta \log n),\] and the overall energy complexity  of the algorithm  is 
\[\sum_{x=0}^{\log K - 1} O(\log \Delta \log n) = O(\log K \log \Delta \log n).\]

\paragraph{Special case.}
For the rest of the proof, we focus on the special case of $\mathcal{S} \cap \mathcal{R} = \emptyset$ and $|\mathcal{R} \cap N(u)| \leq 1$ for each $u\in \mathcal{S}$. These assumptions imply that the family of sets $(\mathcal{S} \cap N(v)) \cup \{v\}$ for all $v \in \mathcal{R}$ are disjoint.
The high-level idea is that for each $v \in \mathcal{R}$, we may let the set of vertices $(\mathcal{S} \cap N(v)) \cup \{v\}$ jointly conduct a binary search to determine all bits of $f_v = \min_{u \in N(v) \cap \mathcal{S}} k_u$, in parallel for all $v \in \mathcal{R}$.

Suppose that for each vertex $v \in \mathcal{R}$, all vertices in the set $(\mathcal{S} \cap N(v)) \cup \{v\}$ already know the first $x$ bits of $f_v$. We present a more efficient algorithm that let all vertices in the set $(\mathcal{S} \cap N(v)) \cup \{v\}$  learn the $(x+1)$th bit of $f_v$.
\begin{description}
\item[Step 1.] Perform $\sr$ with the following choices of $(\mathcal{S}',\mathcal{R}')$:
\begin{itemize}
    \item $\mathcal{R}' = \mathcal{R}$.
    \item $\mathcal{S}'$ is the subset of $\mathcal{S}$ that contains all vertices $u \in \mathcal{S}$ satisfying the following conditions:
\begin{itemize}
    \item The first $x$ bits of $k_u$ equal the first $x$ bits of $f_v$, where $v$ is the unique vertex in $\mathcal{R} \cap N(u)$.
    \item The $(x+1)$th bit of $k_u$ is 0.
\end{itemize}    
\end{itemize}
 This step allows each  $v \in \mathcal{R}$ to learn the  $(x+1)$th bit of $f_v$. If $v \in \mathcal{R}$ receives a message in $\sr$, then $v$ knows that the  $(x+1)$th bit of $f_v$ is 0. Otherwise, $v$ knows that the  $(x+1)$th bit of $f_v$ is 1.

\item[Step 2.] Perform $\sr$ with the following choices of $(\mathcal{S}',\mathcal{R}')$:
\begin{itemize}
    \item $\mathcal{R}' = \mathcal{S}$.
    \item $\mathcal{S}' = \mathcal{R}$.
\end{itemize}
 This step lets each  $v \in \mathcal{R}$  send the  $(x+1)$th bit of $f_v$ to all vertices in $\mathcal{S} \cap N(v)$.
\end{description}
The time and energy complexities of this algorithm are asymptotically the same as that of $\sr$, which are $O(\log \Delta \log n)$. As discussed earlier, to solve $\minsr$, all we need to do is to run the above algorithm from $x=0$ to $x = \log K - 1$.
The overall time and energy complexities of the algorithm for $\minsr$ are $O(\log K \log \Delta \log n)$, as there are $\log K$ iterations.
\end{proof}

For the rest of the section, we consider the task $\apxsr$, which requires each vertex $v \in \mathcal{R}$ to compute a $(1 \pm \epsilon)$-factor approximation of the summation $\sum_{u \in N^+(v) \cap \mathcal{S}} m_u$.
We need the following fact, whose correctness can be verified by means of a simple calculation. 

\begin{fact}\label{lem:aux-cal}
There exist three universal constants $0 < \epsilon_0 < 1$, $N_0 \geq 1$, and $c_0 \geq 1$ such that the following statement holds: For any pair of numbers $(N,\epsilon)$  such that $N \geq N_0$ and  $\epsilon_0 \geq |\epsilon| \geq c_0 / \sqrt{N}$, 
\[e^{-1}(1 - 0.51 \epsilon^2) \leq (1+\epsilon) (1 -  (1+\epsilon)/N)^{N-1} \leq e^{-1}(1 - 0.49 \epsilon^2).\]
\end{fact}

Note that the parameter $\epsilon$ in \cref{lem:aux-cal} can be either positive or negative.
For the rest of the section, we assume that the message $m_u$ sent from each vertex $u \in \mathcal{S}$ is an integer within the range $[W]$.
We first consider the special case of  $\apxsr$ with $W = 1$. In this case, $\apxsr$ is the same as the approximate counting problem whose goal is to let each $v \in \mathcal{R}$ compute  $|N^+(v) \cap \mathcal{S}|$, up to a $(1 \pm \epsilon)$-factor error.

\begin{lemma}\label{lemma:apxsr-1}
For  $W = 1$, $\apxsr$ can be solved in  $O((1/\epsilon^5) \log \Delta \log n)$ time and energy.
\end{lemma}
\begin{proof}
In this proof, we will focus on a slightly different task of estimating $|N(v) \cap \mathcal{S}|$ within a $(1 \pm \epsilon)$-factor approximation, for each $v \in \mathcal{R}$.
If each $v \in \mathcal{R}$ knows such an estimate of $|N(v) \cap \mathcal{S}|$, then $v$ can locally calculate an estimate of $|N^+(v) \cap \mathcal{S}|$ within a $(1 \pm \epsilon)$-factor approximation, thereby solving $\apxsr$ for the case of $W=1$.

\paragraph{Basic setup.}
Let $C > 0$ be a sufficiently large constant.
Let $\epsilon_0, N_0$, and $c_0$ be the constants in \cref{lem:aux-cal}.
We assume that  $\epsilon \leq \epsilon_0$. If this is not the case, then we may reset $\epsilon = \epsilon_0$.

The algorithm consists of two phases. The   first phase of the algorithm aims to achieve the following goals: For each  $v \in \mathcal{R}$, either (i) $v$  learns the number $|N(v) \cap \mathcal{S}|$ exactly or (ii) $v$   detects that  $\epsilon \geq 10 c_0 / \sqrt{|N(v) \cap \mathcal{S}|}$.
For each vertex $v \in \mathcal{R}$ that calculates the number $|N(v) \cap \mathcal{S}|$ exactly in the first phase, we remove $v$ from $\mathcal{R}$.
The second phase of the algorithm then solves $\apxsr$ for the remaining vertices in $\mathcal{R}$. These vertices $v \in \mathcal{R}$ satisfy $\epsilon \geq  10 c_0 / \sqrt{|N(v) \cap \mathcal{S}|}$.

\paragraph{The first phase.}
We define $Z = (10 c_0 / \epsilon)^2$.
The algorithm consists of $C \cdot Z \log n$ rounds, where we do the following in each round:
\begin{itemize}
    \item Each vertex $u \in \mathcal{S} \cup \mathcal{R}$ flips a biased coin that produces head with probability $1/Z$.
    \item Each $u \in \mathcal{S}$ sends $\ID(u)$ if the outcome of its coin flip is head.
    \item Each vertex $v \in \mathcal{R}$  listens if the outcome of its coin flip is tail.
\end{itemize}
For each vertex $v \in  \mathcal{R}$, there are two cases:
\begin{itemize}
    \item Suppose that there is a vertex $u \in N(v) \cap \mathcal{S}$ such that the number of messages that $v$ receives from is smaller than $0.5 \cdot (C \log n)/e$. Then $v$ decides that $\epsilon \geq 10 c_0 / \sqrt{|N(v) \cap \mathcal{S}|}$ and proceeds to the second phase.
    \item Suppose that for all vertices $u \in N(v) \cap \mathcal{S}$, the number of messages that $v$ receives from is at least $0.5 \cdot (C \log n)/e$. Then  $v$ calculate $|N(v) \cap \mathcal{S}|$ by the number of distinct $\ID$s that $v$ receives.
\end{itemize}
The time complexity of the first phase of the algorithm is $C \cdot Z \log n = O((1/\epsilon^{2}) \log n)$.

\paragraph{Analysis.}
To analyze the algorithm, let $e = \{u,v\}$ be any edge such that $u \in \mathcal{S}$ and $v \in \mathcal{R}$. In one round of the above algorithm, $u$ successfully sends a message to $v$ if and only if (i) the outcome of $u$'s coin flip is head, and (ii) the outcome of the coin flips of all vertices in $(N(v)\cap \mathcal{S}) \cup \{v\} \setminus \{u\}$ are all tails. This event occurs with probability $p^\star = (1 - 1/Z)^{|N(v) \cap \mathcal{S}|} \cdot (1/Z)$.
Let $X$ be the number of times $v$ receives a message from $u$.
To prove the correctness of the algorithm, we show the following three concentration bounds:
\begin{itemize}
\item If $v \in  \mathcal{R}$ satisfies $\epsilon \leq  10 c_0 / \sqrt{|N(v) \cap \mathcal{S}|}$, then $\Prob[X \geq 0.8 \cdot (C \log n)/e] = 1 - n^{-\Omega(C)}$.
\item If $v \in  \mathcal{R}$ satisfies $\epsilon \geq  20 c_0 / \sqrt{|N(v) \cap \mathcal{S}|}$, then $\Prob[X \leq 0.2 \cdot (C \log n)/e] = 1 - n^{-\Omega(C)}$.
\item If $v \in  \mathcal{R}$ satisfies $\epsilon \leq  20 c_0 / \sqrt{|N(v) \cap \mathcal{S}|}$, then $\Prob[X \geq 1] = 1 - n^{-\Omega(C)}$.
\end{itemize}

 We show the correctness of the algorithm given these concentration bounds.   For the case  $\epsilon \geq 20 c_0 / \sqrt{|N(v) \cap \mathcal{S}|}$, the second bound implies  that the number of messages that $v$ receives from $u$ is greater than $0.5 \cdot (C \log n)/e$ w.h.p., so $v$ correctly decides that $\epsilon \geq 10 c_0 / \sqrt{|N(v) \cap \mathcal{S}|}$ and proceeds to the second phase.   For the case  $\epsilon \leq 20 c_0 / \sqrt{|N(v) \cap \mathcal{S}|}$, the third bound implies  that $v$ receives at least one message from each vertex in $N(v) \cap \mathcal{S}$ w.h.p., so  $v$ can calculate $|N(v) \cap \mathcal{S}|$ precisely.  The only remaining thing to show is that when $\epsilon$ is at most $10 c_0 / \sqrt{|N(v) \cap \mathcal{S}|}$, w.h.p.~$v$ does not decide that $\epsilon \geq 10 c_0 / \sqrt{|N(v) \cap \mathcal{S}|}$. This follows from the first bound, which implies that the number of messages that $v$ receives from $u$ is greater than $0.5 \cdot (C \log n)/e$ w.h.p.

We prove the three concentration bounds as follows:
\begin{itemize}
    \item Suppose that  vertex $v \in  \mathcal{R}$ satisfies $\epsilon \leq 10 c_0 / \sqrt{|N(v) \cap \mathcal{S}|}$.
We show that in this case the number of messages that $v$ receives from $u \in N(v) \cap \mathcal{S}$ is at least $0.8 \cdot (C \log n)/e$, with probability $1 - n^{-\Omega(C)}$.
In this case, we have $Z = (10 c_0 / \epsilon)^2 \geq |N(v) \cap \mathcal{S}|$, so 
$p^\star = (1 - 1/Z)^{|N(v) \cap \mathcal{S}|} \cdot (1/Z)  \geq  (1 - 1/Z)^{Z} \cdot (1/Z) \geq 0.9 /(e Z)$.
The expected value $\mu$ of $X$ satisfies  $\mu  = C \cdot Z \log n \cdot p^\star \geq  0.9 (C \log n)/e$.
By a Chernoff bound, $\Prob[X \leq 0.8 \cdot (C \log n)/e] \leq \exp(-\Omega(C \log n)) = n^{-\Omega(C)}$.
\item Suppose that   vertex $v \in  \mathcal{R}$ satisfies $\epsilon \geq 20 c_0 / \sqrt{|N(v) \cap \mathcal{S}|}$.
We show that in this case the number of messages that $v$ receives from $u \in N(v) \cap \mathcal{S}$ is at most $0.2 \cdot (C \log n)/e$, with probability $1 - n^{-\Omega(C)}$.
In this case, we have $Z = (10 c_0 / \epsilon)^2 \leq |N(v) \cap \mathcal{S}| / 4$, so
$p^\star = (1 - 1/Z)^{|N(v) \cap \mathcal{S}|} \cdot (1/Z) \leq  (1 - 1/Z)^{4Z} \cdot (1/Z) \leq 1/(e^4 Z)$.
The expected value $\mu$ of $X$ satisfies  $\mu  = C \cdot Z \log n \cdot p^\star \leq  (C \log n)/ e^4 < 0.1 (C \log n)/ e$.
By a Chernoff bound, $\Prob[X \geq 0.2 \cdot (C \log n)/e] \leq \exp(-\Omega(C \log n)) = n^{-\Omega(C)}$.
\item Suppose that  vertex $v \in  \mathcal{R}$ satisfies $\epsilon \leq  20c_0 / \sqrt{|N(v) \cap \mathcal{S}|}$.
We show that in this case the number of messages that $v$ receives from $u \in N(v) \cap \mathcal{S}$ is at least $1$, with probability $1 - n^{-\Omega(C)}$.
In this case, we have $Z = (10 c_0 / \epsilon)^2 \geq |N(v) \cap \mathcal{S}| / 4$, so
$p^\star = (1 - 1/Z)^{|N(v) \cap \mathcal{S}|} \cdot (1/Z) \geq  (1 - 1/Z)^{4Z} \cdot (1/Z) \geq  0.9 /(e^4 Z)$.
We have $\Prob[X < 1] = (1-p^\star)^{CZ \log n} \leq  (1 - 0.9 /(e^4 Z))^{CZ \log n}   = n^{-\Omega(C)}$.
\end{itemize}

\paragraph{The second phase.}
For each vertex $v \in \mathcal{R}$ that have
 already calculated the number $|N(v) \cap \mathcal{S}|$ exactly in the first phase,   $v$ removes itself from $\mathcal{R}$.
We know that all the remaining vertices in $\mathcal{R}$ satisfy $\epsilon \geq  10 c_0 / \sqrt{|N(v) \cap \mathcal{S}|}$.
 
We consider the sequence of sending probabilities: $p_1 = 2/\Delta$, and $p_i = \min\{1, p_{i-1} \cdot (1+ \epsilon)\}$ for $i > 1$. We let $i^\star = O((1/\epsilon) \log \Delta)$ be the smallest index $i$ such that $p_i = 1$.

The second phase of the algorithm consists of $i^\star$ iterations, where the $i$th iteration repeats the following procedure for $C \cdot (1/\epsilon^4) \log n$ times for all  vertices $v \in \mathcal{S} \cup \mathcal{R}$:
\begin{itemize}
    \item $v$ flips a fair coin.
    \item If the outcome of the coin flip is head and $v \in \mathcal{S}$, then $v$ sends with probability $p_i$.
    \item If the outcome of the coin flip is tail and  $v \in \mathcal{R}$, then $v$ listens to the channel.
\end{itemize}

After finishing the algorithm, each vertex $v \in \mathcal{R}$ finds an index $i'$ such that the number of messages that $v$ successfully receives during the $i'$th iteration is the highest. Then $v$ decides that $2 / p_{i'}$ is an  estimate of $|N(v) \cap \mathcal{S}|$ within a factor of $(1\pm \epsilon)$.
The time complexity of the second phase of the algorithm is $i^\star \cdot C \cdot (1/\epsilon^4) \log n = O((1/\epsilon^5) \log \Delta \log n)$.

\paragraph{Analysis.} To show the correctness of the above algorithm, in the subsequent discussion, we focus on a vertex $v \in \mathcal{R}$ in the $i$th iteration.
We say that $i$ is {\em good} for $v$ if $p_i/2$ is within a  $(1\pm 0.6\epsilon)$-factor of $1/|N(v) \cap \mathcal{S}|$,
and we say that $i$ is {\em bad} for $v$ if $p_i/2$ is not within a  $(1\pm \epsilon)$-factor of $1/|N(v) \cap \mathcal{S}|$.
Our choice of the sequence $(p_1, p_2, \ldots)$ implies that there must be at least one good index $i$ for $v$.

We write  $p_i^{\text{suc}}$ to denote the probability that $v$ successfully receives a message in one round of the $i$th iteration.
From the description of the algorithm, we have \[p_i^{\text{suc}} = (1/2) \cdot |N(v) \cap \mathcal{S}| \cdot (p_i / 2) \cdot (1-(p_i / 2))^{|N(v) \cap \mathcal{S}|-1}.\]
We define
\[ p_{\text{good}} = (1/2) \cdot e^{-1}(1 - 0.51 (0.6 \epsilon)^2) \ \ \ \text{and}  \ \ \ p_{\text{bad}} = (1/2) \cdot e^{-1}(1 - 0.49 \epsilon^2).\]
We claim that (i) $p_i^{\text{suc}} \geq p_{\text{good}}$ if $i$ is good for $v$ and (ii) $p_i^{\text{suc}} \leq p_{\text{bad}}$ if $i$ is bad for $v$.

We first prove this claim for the case that $i$ is good for $v$. For simplicity, we write $N = |N(v) \cap \mathcal{S}|$. Since $i$ is good, $p_i / 2 = (1 + \epsilon')/|N(v) \cap \mathcal{S}|$ for some $\epsilon' \in [-0.6 \epsilon, 0.6 \epsilon]$. Using the new notations, we may rewrite  $p_i^{\text{suc}}$
as 
\[p_i^{\text{suc}} = (1/2) \cdot |N(v) \cap \mathcal{S}| \cdot (p_i / 2) \cdot (1-(p_i / 2))^{|N(v) \cap \mathcal{S}|-1}
= (1/2) \cdot (1 + \epsilon') \cdot (1-(1 + \epsilon'))^{N-1}.
\]
By \cref{lem:aux-cal}, we infer that $p_i^{\text{suc}} \geq  (1/2) \cdot  e^{-1}(1 - 0.51 (\epsilon')^2) \geq e^{-1}(1 - 0.51 (0.6 \epsilon)^2) = p_{\text{good}}$.

Now consider the case $i$ is bad for $v$. Again, we write $N = |N(v) \cap \mathcal{S}|$. Since $i$ is bad, $p_i / 2 = (1 + \epsilon')/|N(v) \cap \mathcal{S}|$ for some $\epsilon' \notin (-\epsilon, \epsilon)$. The above formula for $p_i^{\text{suc}}$ still applies to this case, and \cref{lem:aux-cal} implies that $p_i^{\text{suc}} \leq  (1/2) \cdot  e^{-1}(1 - 0.49 (\epsilon')^2) \leq e^{-1}(1 - 0.49 \epsilon^2) = p_{\text{bad}}$.

Let $X$ be the number of messages that $v$ receives in the $i$th iteration of the algorithm. The expected value of $X$ is $\mu = p_i^{\text{suc}} \cdot C  \cdot (1/\epsilon^4) \log n$.  For the case $i$ is good for $v$, we have $\mu \geq p_{\text{good}} \cdot C  \cdot(1/\epsilon^4) \log n$, so by a Chernoff bound, we have:
\[\Prob[X \leq (1-0.01 \epsilon^2) p_{\text{good}}  \cdot C \cdot(1/\epsilon^4) \log n] =  \exp(-\Omega(\epsilon^4 \cdot C \cdot(1/\epsilon^4) \log n)) = n^{-\Omega(C)}.\]
For the case $i$ is bad for $v$, we have  $\mu \leq p_{\text{bad}} \cdot C  \cdot (1/\epsilon^4) \log n$, so by a Chernoff bound, we have:
\[\Prob[X \geq (1+0.01 \epsilon^2) p_{\text{bad}}  \cdot C \cdot(1/\epsilon^4) \log n] =  \exp(-\Omega(\epsilon^4 \cdot C \cdot(1/\epsilon^4) \log n)) =  n^{-\Omega(C)}.\]

Since $(1-0.01 \epsilon^2) p_{\text{good}} > (1+0.01 \epsilon^2) p_{\text{bad}}$, we conclude that w.h.p.~the index $i'$ selected by $v$ must be good, which implies that the estimate $2/p_{i'}$ calculated by $v$ is within a $(1\pm \epsilon)$-factor of $|N(v) \cap \mathcal{S}|$, as we know that $p_{i'}/2$ is within a $(1\pm 0.6\epsilon)$-factor of $1/|N(v) \cap \mathcal{S}|$, as $i'$ is good.
\end{proof}

In the following lemma, we extend \cref{lemma:apxsr-1} to any value of $W$.

\begin{lemma}\label{lemma:apxsr}
 $\apxsr$ can be solved in   $O((1/\epsilon^6) \log W \log \Delta \log n)$ time and energy.
\end{lemma}
\begin{proof}
We let $\epsilon' = \Theta(\epsilon)$ be chosen such that $(1+\epsilon')^2 < 1+\epsilon$ and $(1-\epsilon')^2 > 1-\epsilon$.
We consider the following sequence:  $w_1 = 1$ and $w_i = \min\{ W, (1+\epsilon')w_{i-1}\}$ for $i > 1$.
Let $i^\star$ be the smallest index $i$ such that $w_i = W$.

From $i = 1$ to $i^\star$, we run the algorithm of \cref{lemma:apxsr-1} with the following setting: 
\begin{itemize}
    \item  $\mathcal{S}'$ is the vertices $u \in \mathcal{S}$ with $m_u \in (w_{i-1}, w_{i}]$.
    \item  $\mathcal{R}' = \mathcal{R}$.
    \item The error parameter is $\epsilon'$.
\end{itemize}
The algorithm of \cref{lemma:apxsr-1} lets each $v \in \mathcal{R}'$ compute a $(1 \pm \epsilon')$-factor approximation of $|N^+(v) \cap \mathcal{S}'|$ using $O((1/\epsilon^5) \log \Delta \log n)$ time and energy.

For each $v \in \mathcal{R}$, we write $N_i$ to denote the number of vertices $u \in N^+(v) \cap \mathcal{S}$ such that  $m_u \in (w_{i-1}, w_{i}]$, and we write $\tilde{N}_i$ to denote the estimate of $|N^+(v) \cap \mathcal{S}'|$ computed by $v$ in the $i$th iteration. We have the following observations:
\begin{itemize}
    \item $\tilde{N}_i$  is a $(1\pm \epsilon')$-factor approximation of $N_i$.
    \item  $\sum_{i=1}^{i^\star} w_i N_i$ is a $(1\pm \epsilon')$-factor approximation of $\sum_{u \in N^+(v) \cap \mathcal{S}} m_u$.
\end{itemize}
Thus, $\sum_{i=1}^{i^\star} w_i \tilde{N}_i$, which can be calculated locally at $v$ at the end of the algorithm, is a $(1\pm \epsilon)$-factor approximation of $\sum_{u \in N^+(v) \cap \mathcal{S}} m_u$, by our choice of $\epsilon'$.

By \cref{lemma:apxsr-1}, the time and energy complexities for each iteration are $O((1/\epsilon^5)\log \Delta  \log n)$. The total number of iterations is $i^\star = O((1/\epsilon) \log W)$. Thus, the overall time and energy complexities are $O((1/\epsilon^6) \log W \log \Delta \log n)$.
\end{proof} 

\end{document}